\documentclass[a4paper,UKenglish,cleveref, autoref]{lipics-v2019}
\usepackage{amsmath,amsthm,amssymb}
\usepackage{xspace}
\usepackage{tikz}

\title{Parameterized Approximation Schemes for Independent Set of Rectangles and Geometric Knapsack} %TODO Please add

\titlerunning{Parameterized Approximation Schemes}%optional, please use if title is longer than one line

\author{Fabrizio Grandoni}{IDSIA, USI-SUPSI, Switzerland}{fabrizio@idsia.ch}{https://orcid.org/0000-0002-9676-4931}{%
}
\author{Stefan Kratsch}{Institut f\"ur Informatik, Humboldt Universit\"at zu Berlin, Germany}{kratsch@informatik.hu-berlin.de}{https://orcid.org/0000-0002-0193-7239}{%
}

\author{Andreas Wiese}{Department of Industrial Engineering and Center for Mathematical Modeling, Universidad de Chile, Chile}{awiese@dii.uchile.cl}{%
}{%
Partially supported by the Fondecyt Regular grant 1170223.
}

\authorrunning{F. Grandoni, S. Kratsch, and A. Wiese}%TODO mandatory. First: Use abbreviated first/middle names. Second (only in severe cases): Use first author plus 'et al.'

\Copyright{Fabrizio Grandoni, Stefan Kratsch, and Andreas Wiese}%TODO mandatory, please use full first names. LIPIcs license is "CC-BY";  http://creativecommons.org/licenses/by/3.0/

\ccsdesc[500]{Theory of computation~Packing and covering problems}
\ccsdesc[300]{Theory of computation~Fixed parameter tractability}
\keywords{parameterized approximation, parameterized intractability, independent set of rectangles, geometric knapsack}%TODO mandatory; please add comma-separated list of keywords

\category{}%optional, e.g. invited paper

\relatedversion{}%optional, e.g. full version hosted on arXiv, HAL, or other respository/website
\supplement{}%optional, e.g. related research data, source code, ... hosted on a repository like zenodo, figshare, GitHub, ...

\nolinenumbers %uncomment to disable line numbering

\newcommand{\FPT}{\ensuremath{\mathsf{FPT}}\xspace}
\newcommand{\W}[1]{\ensuremath{\mathsf{W[#1]}}\xspace}

\newcommand{\N}{\mathbb{N}}
\newcommand{\R}{\mathbb{R}}

\newcommand{\probname}[1]{\lowercase{\textsc{#1}}}
\newcommand{\problem}[1]{\probname{#1}\xspace}

\newcommand{\SSum}{\problem{Subset Sum}}
\newcommand{\MSS}{\problem{Multi-Subset Sum}}
\newcommand{\TDK}{\problem{2DK}}
\newcommand{\TDKR}{\problem{2DKR}}
\newcommand{\MISR}{\problem{MISR}}

\newcommand{\LEFT}{\mathtt{left}}
\newcommand{\RIGHT}{\mathtt{right}}
\newcommand{\TOP}{\mathtt{top}}
\newcommand{\BOTTOM}{\mathtt{bottom}}
\newcommand{\HEIGHT}{\mathtt{height}}
\newcommand{\WIDTH}{\mathtt{width}}

\newcommand{\eps}{\varepsilon}

\global\long\def\sp{B}
\global\long\def\R{\mathcal{R}}
\global\long\def\L{\mathcal{L}}
\global\long\def\G{\mathcal{G}}
\global\long\def\C{\mathcal{C}}
\global\long\def\OPT{\mathrm{OPT}}

\begin{document}

\maketitle

%TODO mandatory: add short abstract of the document
\begin{abstract}
The area of parameterized approximation seeks to combine approximation and parameterized algorithms to obtain, e.g., $(1+\eps)$-approximations in $f(k,\eps)n^{O(1)}$ time where $k$ is some \emph{parameter} of the input. The goal is to overcome lower bounds from either of the areas. We obtain the following results on parameterized approximability:
 \begin{itemize}\itemsep0pt
 \item[$\bullet$] In the \emph{maximum independent set of rectangles} problem (\MISR) we are given a collection of $n$ axis parallel rectangles in the plane. Our goal is to select a maximum-cardinality subset of pairwise non-overlapping rectangles. This problem is NP-hard and also \W{1}-hard
 [Marx, ESA'05]. 
 The best-known polynomial-time approximation factor is $O(\log\log n)$ 
 [Chalermsook and Chuzhoy, SODA'09] 
 and it admits a QPTAS 
 [Adamaszek and Wiese, FOCS'13; Chuzhoy and Ene, FOCS'16].
 Here we present a \emph{parameterized approximation scheme} (PAS) for \MISR, i.e. an algorithm that, for any given constant $\eps>0$ and integer $k>0$, in time $f(k,\eps)n^{g(\eps)}$, either outputs a solution of size at least
 $k/(1+\eps)$, 
 or declares that the optimum solution has size less than $k$. 
 \item[$\bullet$] In the \emph{(2-dimensional) geometric knapsack} problem (\TDK) we are given an axis-aligned square knapsack and a collection of axis-aligned rectangles in the plane (\emph{items}). Our goal is to translate a maximum cardinality subset of items into the knapsack so that the selected items do not overlap. In the version of \TDK with rotations (\TDKR), we are allowed to rotate items by 90 degrees. Both variants are NP-hard, and the best-known polynomial-time approximation factor is $2+\eps$
 [Jansen and Zhang, SODA'04]. 
 These problems admit a QPTAS for polynomially bounded item sizes
 [Adamaszek and Wiese, SODA'15]. 
 We show that both variants are \W{1}-hard. Furthermore, we present a PAS for \TDKR.
 \end{itemize}
 For all considered problems, getting time $f(k,\eps)n^{O(1)}$, rather than $f(k,\eps)n^{g(\eps)}$, would give FPT time $f'(k)n^{O(1)}$ exact algorithms by setting $\eps=1/(k+1)$, contradicting \W{1}-hardness. Instead, for each fixed $\eps>0$, our PASs give $(1+\eps)$-approximate solutions in FPT time.
 
 For both \MISR and \TDKR our techniques also give rise to preprocessing algorithms that take $n^{g(\eps)}$ time and return a subset of at most $k^{g(\eps)}$ rectangles/items that contains a solution of size at least 
 $k/(1+\eps)$ 
 if a solution of size $k$ exists. This is a special case of the recently introduced notion of a polynomial-size approximate kernelization scheme [Lokshtanov et al., STOC'17].
\end{abstract}

\section{Introduction}\label{section:introduction}

Approximation algorithms and parameterized algorithms are two well-established ways to deal with NP-hard problems. An $\alpha$-approximation for an optimization problem is a polynomial-time algorithm that computes a \emph{feasible} solution whose cost is within a factor $\alpha$ (that might be a function of the input size $n$) of the optimal cost. In particular, a \emph{polynomial-time approximation scheme} (PTAS) is a $(1+\eps)$-approximation algorithm running in time $n^{g(\eps)}$, where $\eps>0$ is a given constant and $g$ is some computable function.
In parameterized algorithms we identify a parameter $k$ of the input, that we informally assume to be much smaller than $n$. The goal here is to solve the problem optimally in \emph{fixed-parameter tractable} (FPT) time $f(k)n^{O(1)}$, where $f$ is some computable function.
Recently, researchers started to combine the two notions (see, e.g., the survey by Marx~\cite{Marx08}). The idea is to design approximation algorithms that run in FPT (rather than polynomial) time, e.g., to get $(1+\eps)$-approximate solutions in time $f(k,\eps)n^{O(1)}$. In this paper we continue this line of research on parameterized approximation, and apply it to two fundamental rectangle packing problems.

\subsection{Our results and techniques}

Our focus is on parameterized approximation algorithms. Unfortunately, as observed by Marx~\cite{Marx08}, when the parameter $k$ is the desired solution size, computing $(1+\eps)$-approximate solutions in time $f(k,\eps)n^{O(1)}$ implies fixed-parameter tractability. Indeed, setting $\eps=1/(k+1)$ guarantees to find an optimal solution when that value equals to $k \in \N$ and we get time $f(k,1/(k+1))n^{O(1)}=f'(k) n^{O(1)}$. Since the considered problems are \W{1}-hard (in part, this is established in our work), they are unlikely to be FPT and similarly unlikely to have such nice approximation schemes.

Instead, we construct algorithms (for two maximization problems) that, given $\eps>0$ and an integer $k$, take time $f(k,\eps)n^{g(\eps)}$ and either return a solution of size at least $k/(1+\eps)$ or declare that the optimum is less than $k$. 
We call such an algorithm a \emph{parameterized approximation scheme (PAS)}.
Note that if we run such an algorithm for each $k'\le k$ then we can guarantee that we compute a solution with cardinality at least $\min\{k,\OPT\}/(1+\eps)$ where $\OPT$ denotes the size of the optimal solution.
So intuitively, for each $\eps>0$, we have an FPT-algorithm for getting a $(1+\eps)$-approximate solution.

In this paper we consider the following two geometric packing problems, and design PASs for them.

\subparagraph{Maximum Independent Set of Rectangles.} In the \emph{maximum independent set of rectangles} problem (\MISR) we are given a set of $n$ axis-parallel rectangles $\R=\{R_{1},\dots,R_{n}\}$
in the two-dimensional plane, where $R_{i}$ is the open set of points $(x_{i}^{(1)},x_{i}^{(2)})\times (y_{i}^{(1)},y_{i}^{(2)})$.
A feasible solution is a subset of rectangles $\R'\subseteq\R$ such
that for any two rectangles $R,R'\in\R'$ we have $R\cap R'=\emptyset$.
Our objective is to find a feasible solution of maximum cardinality  $|\R'|$. 
W.l.o.g. we assume that $x_{i}^{(1)},y_{i}^{(1)},x_{i}^{(2)},y_{i}^{(2)}\in\{0,\dots,2n-1\}$
for each $R_{i}\in\R$ 
(see e.g. \cite{adamaszek2013approximation}).

\MISR is very well-studied in the area of approximation algorithms. The problem is known to be NP-hard \cite{fowler1981optimal}, and the current best polynomial-time approximation factor is $O(\log\log n)$ for the cardinality case \cite{CC2009} (addressed in this paper), and $O(\log n/\log\log n)$ for the natural generalization with rectangle weights \cite{chan2012approximation}. 
The cardinality case also admits a $(1+\eps)$-approximation with a running time of $n^{poly(\log \log (n/\eps))}$ \cite{chuzhoy2016approximating} and there is a (slower) QPTAS known for the weighted case~\cite{adamaszek2013approximation}.
The problem is also known to be \W{1}-hard w.r.t. the number $k$ of rectangles in the solution \cite{Marx05}, and thus unlikely to be solvable in FPT time $f(k)n^{O(1)}$.

In this paper we achieve the following main result:
\begin{theorem}
\label{thm:MISR-FPT-PTAS} There is a PAS for \MISR with running time $k^{O(k/\epsilon^{8})}n^{O(1/\epsilon^{8})}$.
\end{theorem}
In order to achieve the above result, we combine several ideas. Our starting point is a polynomial-time construction of a $k\times k$ grid such that each rectangle in the input contains some crossing point of this grid (or we find a solution of size $k$ directly). By applying (in a non-trivial way) a
result by Frederickson~\cite{federickson1987fast} on planar graphs, and losing a small factor in the approximation, we define a decomposition of our grid into a collection of disjoint \emph{groups} of cells. Each such group defines an independent instance of the problem, consisting of the rectangles strictly contained in the considered group of cells. Furthermore, we guarantee that each group spans only a \emph{constant number} $O_\eps(1)$ of rectangles of the optimum solution. Therefore in FPT time we can guess the correct decomposition, and solve each corresponding subproblem in $n^{O_\eps(1)}$ time. We remark that our approach deviates substantially from prior work, and might be useful for other related problems.

An adaptation of our construction also leads to the following $(1+\epsilon)$-approximative kernelization.
 
\begin{theorem}\label{thm:MISR-FPT-PTAS-kernel}
There is an algorithm
for \MISR that, given $k\in \N$, computes in time $n^{O(1/\epsilon^{8})}$ a subset of the input rectangles of size $k^{O(1/\epsilon^{8})}$ that
contains a solution of size at least $k/(1+\eps)$, assuming that the input instance admits a solution of size at least $k$.
\end{theorem}
Similarly as for a PAS, if we run the above algorithm for each $k'\le k$ we obtain a set of size $k^{O(1/\epsilon^{8})}$ that contains a 
solution of size at least $\min\{k,\OPT\}/(1+\eps)$. 
Observe that any $c$-approximate solution on the obtained set of rectangles is also a feasible, and $c(1+\eps)$-approximate, solution for the original instance if $\OPT\leq k$ and otherwise has size at least $k/(c(1+\eps))$. Thus, our result is a special case of a \emph{polynomial-size approximate kernelization scheme} (PSAKS) as defined in \cite{LokshtanovPRS17}.\footnote{The definition due to Lokshtanov et al.~\cite{LokshtanovPRS17} is not restricted to generating a small subset of the input and a dedicated solution lifting algorithm may be used.}

\subparagraph{2-Dimensional Geometric Knapsack.}

In the \emph{(2-Dimensional) Geometric Knapsack} problem ($\TDK$) we are given 
a square \emph{knapsack} $[0,N]\times [0,N]$, $N\in \mathbb{N}$, and a set of $n$ items $I$, where each item $i\in I$ is an open rectangle $(0,w_i)\times (0,h_i)$, $N\geq w_i,h_i\in \mathbb{N}$.
The goal is to find a feasible \emph{packing} of a subset
$I'\subseteq I$ of the items of maximum cardinality $|I'|$. Such packing maps each item $i\in I'$ into a new translated rectangle $(a_i,a_i+w_i)\times (b_i,b_i+h_i)$\footnote{Intuitively, $i$ is shifted by $a_i$ to the right and by $b_i$ to the top.}, so that the translated rectangles are fully contained in the knapsack and do not overlap with each other.
Here we also consider a variant of \TDK \emph{with rotations} (\TDKR) where we can rotate each input rectangle by 90 degrees.

Both \TDK and \TDKR are NP-hard \cite{leung1990packing} and admit a polynomial-time $(2+\eps)$-approximation for any constant $\eps>0$ \cite{jansen2004rectangle}. These problems admit a QPTAS if $N=n^{O(1)}$ \cite{Adamaszek2015}. Somewhat surprisingly, these problems are not known to be \W{1}-hard when parameterized by the output number $k$ of items. Note that showing \W{1}-hardness is important in our case to motivate the search for a PAS.

\begin{theorem}\label{thm:w1hard}
\TDK and \TDKR are \W{1}-hard when parameterized by $k$.
\end{theorem}

The result is proved by parameterized reductions from a variant of the \W{1}-hard \SSum problem, where we need to determine whether a set of $m$ positive integers contains a $k$-tuple of numbers with sum equal to some given value $t$. The difficulty for reductions to \TDK or \TDKR is of course that rectangles may be freely selected and placed (and possibly rotated) to get a feasible packing.

We complement the \W{1}-hardness result by giving a PAS for the case with rotations (\TDKR) and a corresponding kernelization procedure like in Theorem~\ref{thm:MISR-FPT-PTAS-kernel} (which also yields a PSAKS).
\begin{theorem}\label{thr:tdkr:pas} 
For \TDKR there is a PAS with running time $k^{O(k/\epsilon)}n^{O(1/\epsilon^{3})}$
and an algorithm
that, given $k\in \N$, computes in time $n^{O(1/\epsilon^{3})}$  a subset of the input items of size  $k^{O(1/\epsilon)}$ that
contains a solution of size at least $k/(1+\eps)$, assuming that the input instance admits a solution of size at least $k$.
\end{theorem}
The above result is based on a simple combination of the following two (non-trivial) building blocks:
First, we show that, by losing a fraction $\eps$ of the items of a given solution of size $k$, it is possible to free a vertical strip of width $N/k^{O_\eps(1)}$ (unless the problem can be solved trivially). This is achieved by 
first sparsifying the solution using the above mentioned result by Frederickson~\cite{federickson1987fast}.
If this is not sufficient we construct
a \emph{vertical chain} of relatively wide and tall rectangles that split the instance into a left and right side.   
Then we design a \emph{resource augmentation} algorithm, however in an FPT sense: we can compute in FPT time a packing of cardinality $k$ if we are allowed to use a knapsack where one side is enlarged by a factor $1+1/k^{O_\eps(1)}$. Note that in typical resource augmentation results the packing constraint is relaxed by a constant factor while here this amount is controlled by our parameter.

\subsection{Related work}

One of the first fruitful connections between parameterized complexity and approximability was observed independently by Bazgan~\cite{Bazgan1995} and Cesati and Trevisan~\cite{CesatiT1997}: They showed that EPTASs, i.e., $(1+\eps)$-approximation algorithms with $f(\eps)n^{O(1)}$ time, imply fixed-parameter tractability for the decision version. Thus, proofs for \W{1}-hardness of the decision version became a strong tool for ruling out improvements of PTASs, with running time $n^{g(\eps)}$, to EPTASs. More recently, Boucher et al.~\cite{BoucherLL15} improved this approach by directly proving \W{1}-hardness of obtaining a $(1+\eps)$-approximation, thus bypassing the requirement of a \W{1}-hard decision version (see also~\cite{CyganLPPS16}).

The systematic study of parameterized approximation as a field was initiated independently by three separate publications~\cite{CaiH06,ChenGG06,DowneyFM06}. A very good introduction to the area including key definitions as well as a survey of earlier results that fit into the picture was given by Marx~\cite{Marx08}. In particular, Marx also defined a so-called \emph{standard FPT-approximation algorithm (with performance ratio $c$)} that, given input $(x,k)$ will run for $f(k)|x|^{O(1)}$ time and return (say, for a maximization problem) a solution of value at least $k/c$ if the optimum is at least $k$. As mentioned earlier, Marx pointed out that a standard FPT-approximation scheme
that finds a solution of value at least $k/(1+\eps)$ in time $f(k,\eps)|x|^{O(1)}$ if $\OPT\geq k$ is not interesting to study: By setting $\eps=1/(k+1)$ we can decide the decision problem ``$\OPT\geq k$?'' in FPT time. Thus, such a scheme is not helpful if the decision problem is \W{1}-hard and therefore unlikely to have an FPT-algorithm.
Nevertheless, PASs can be useful in this case, as they imply standard FPT-approximation algorithms with ratio $1+\eps$ for each fixed $\eps>0$ despite \W{1}-hardness.

A central goal of parameterized approximation is to settle the status of problems like \problem{Dominating Set} or \problem{Clique}, which are hard to approximate and also parameterized intractable. Recently, Chen and Lin~\cite{ChenL16} made important progress by showing that \problem{Dominating Set} admits no constant-factor approximation with running time $f(k)n^{O(1)}$ unless $\FPT=\W{1}$. Generally, for problems without exact FPT-algorithms, the goal is to find out whether one can beat inapproximability bounds by allowing FPT-time in some parameter; see e.g.~\cite{FellowsKRS12,BazganCNS14a,BazganCNS14b,BazganN14,Lampis14,KolayMRS14,Cohen-AddadM15,Feldmann15,BodlaenderDDFLP16}).

For the special case of \MISR where all input objects are squares a PTAS is known \cite{erlebach2005polynomial} but there can be no EPTAS~\cite{Marx05}. 
Recently,
Galvez et al. \cite{Galvez2017} found polynomial-time algorithms for \TDK and \TDKR with approximation ratio smaller than $2$ (also for the weighted case).
For the special case that all input objects are squares there is a PTAS \cite{jansen2008polynomial} and even an EPTAS~\cite{heydrich2017faster}.

\section{A Parameterized Approximation Scheme for MISR}
\label{sec:misr}

In this section we present a PAS and an approximate kernelization
for \MISR. We start by showing that there exists an almost optimal solution for the problem with some helpful structural properties (Sections \ref{subsec:grid} and \ref{subsec:Groups-of-rectangles}). 
The results are then put together in Section~\ref{sec:pas}.

\subsection{\label{subsec:grid}Definition of the grid}

We try to construct a non-uniform grid with $k$ rows and $k$ columns
such that each input rectangle overlaps a corner of this grid (see
Figure~\ref{fig:grid}). To this end, we want to compute $k-1$ vertical
and $k-1$ horizontal lines such that each input rectangle intersects
one line from each set. There are instances in which our routine fails
to construct such a grid (and in fact such a grid might not even exist).
For such instances, we directly find a feasible solution with $k$
rectangles and we are done.

\begin{figure}[t]
\begin{centering}
\includegraphics[scale=0.39]{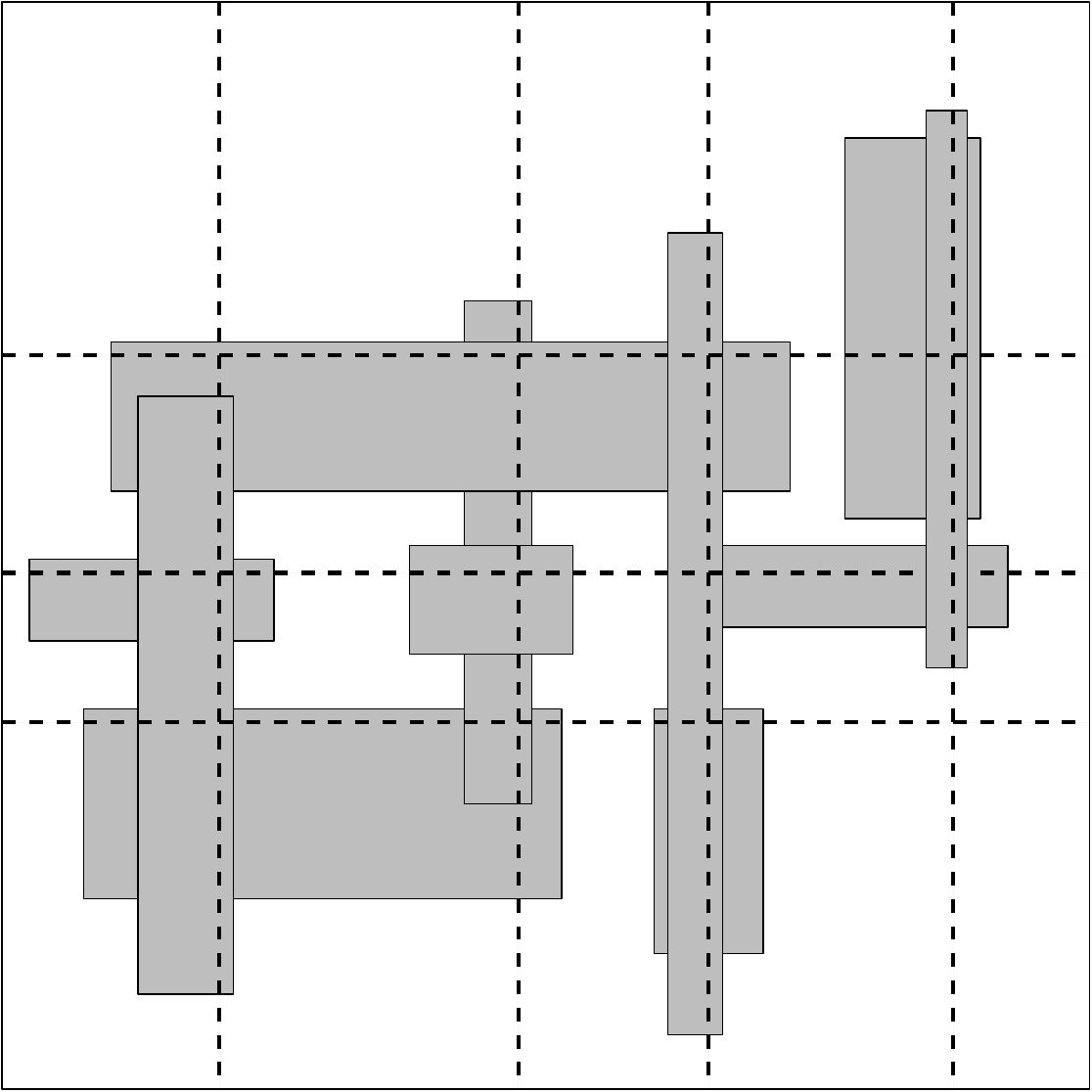}~~~\includegraphics[scale=0.39]{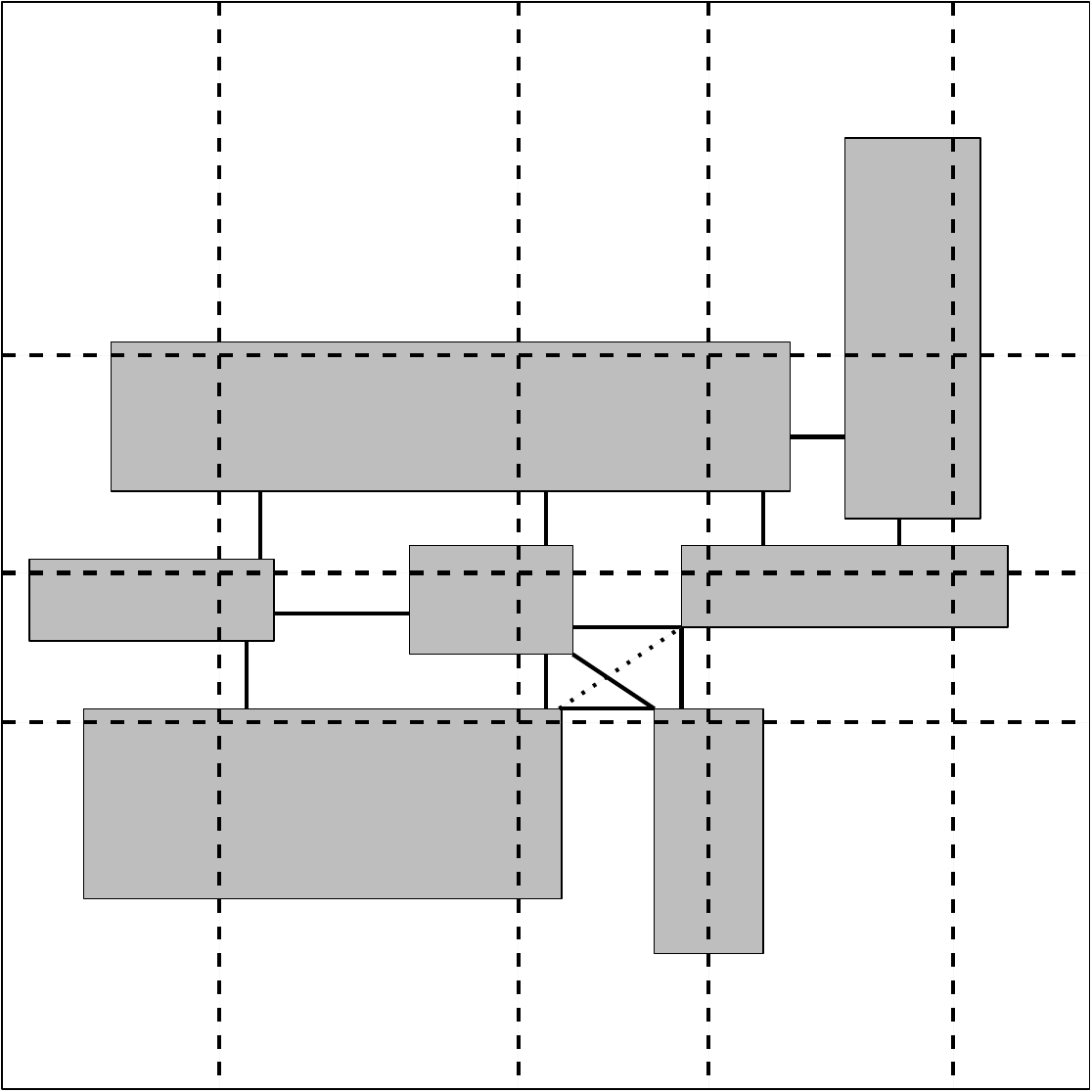}~~~\includegraphics[scale=0.39]{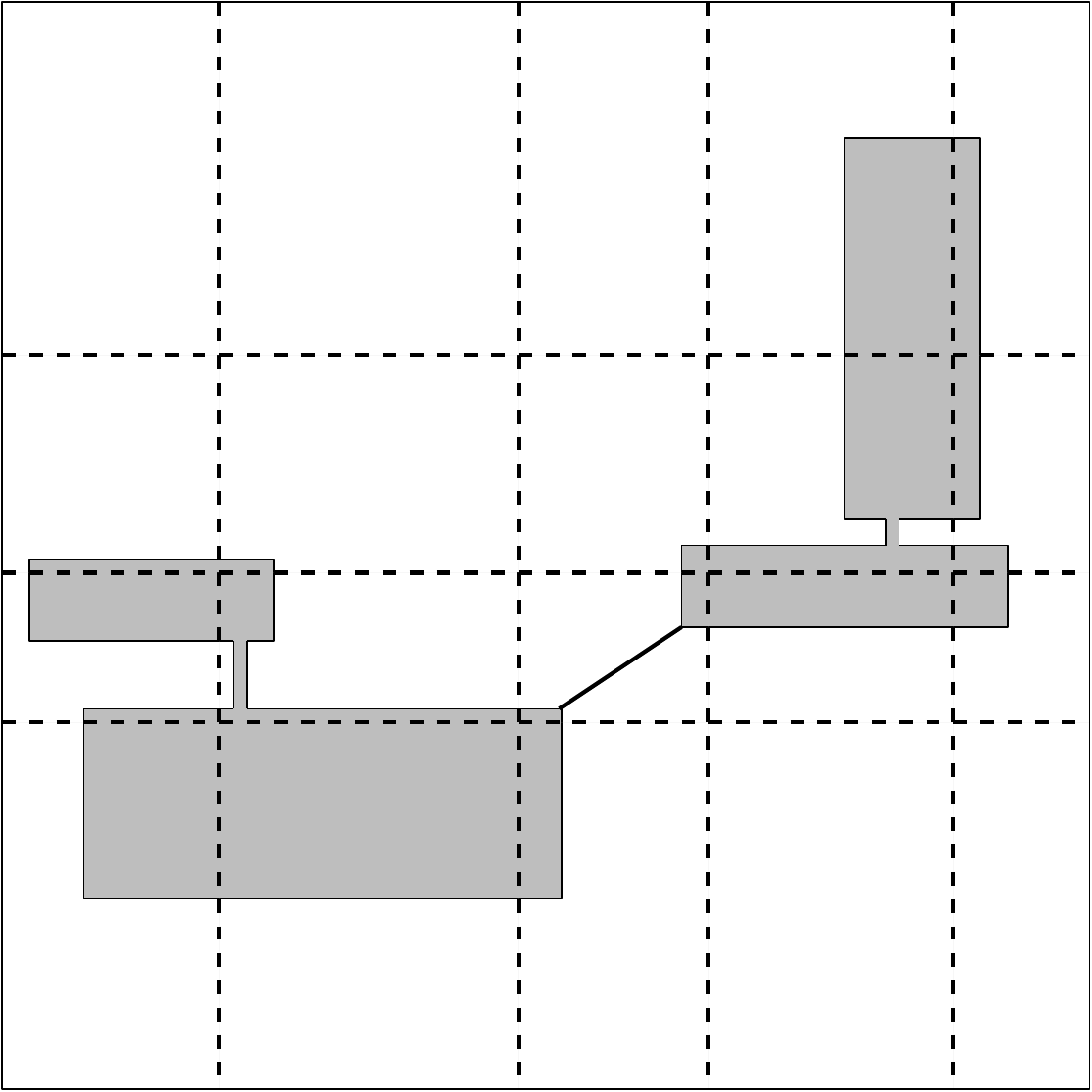}
\par\end{centering}
\caption{\label{fig:grid} (Left) Dashed lines define the grid $\protect\G$.
(Middle) Rectangles from an optimal solution and the
edges that form the graph $G_{1}$. Note that in $G_{1}$ there is
no edge representing the dotted connection since otherwise the graph
would not be planar anymore. (Right) The graph $G_2$, that captures the missing connections of $G_1$}
\end{figure}

\begin{lemma} \label{lem:grid}There is a polynomial time algorithm
that either computes a set of at most $k-1$ vertical lines $\L_{V}$
with $x$-coordinates $\ell_{1}^{V},\dots,\ell_{k-1}^{V}$ such
that each input rectangle is crossed by one line in $\L_{V}$ or computes
a feasible solution with $k$ rectangles. A symmetric statement holds
for an algorithm computing a set of at most $k-1$ horizontal lines
$\L_{H}$ with $y$-coordinates $\ell_{1}^{H},\dots,\ell_{k-1}^{H}$.
\end{lemma}

\begin{proof}
Let $\ell_{0}^{V}:=0$. Assume
inductively that we defined the $x$-coordinates $\ell_{0}^{V},\ell_{1}^{V},\dots,\ell_{k'}^{V}$
such that $\ell_{1}^{V},\dots,\ell_{k'}^{V}$ are the $x$-coordinates
of the first $k'$ constructed vertical lines. We define the $x$-coordinate
of the $(k'+1)$-th vertical line by $\ell_{k'+1}^{V}:=\min_{R_{i}\in\R:x_{i}^{(1)}
\ge
\ell_{k'}^{V}}x_{i}^{(2)}-1/2$.
We continue with this construction until we reach an iteration $k^{*}$
such that $\{R_{i}\in\R:x_{i}^{(1)}
\ge
\ell_{k^{*}-1}^{V}\}=\emptyset$.
If $k^{*}\le k$ then we constructed at most $k-1$ lines such that
each input rectangle is intersected by one of these lines. Otherwise,
assume that $k^{*}>k$. Then for each iteration $k'\in\{1,\dots,k\}$
we can find a rectangle $R_{i(k')}:=\arg\min_{R_{i}\in\R:x_{i}^{(1)}
\ge
\ell_{k'-1}^{V}
}
x_{i}^{(2)}$. By construction, using the fact that all coordinates are integer, for any two such rectangles $R_{i(k')},R_{i(k'')}$
with $k'\ne k''$ we have that $(x_{i(k')}^{(1)},x_{i(k')}^{(2)})\cap(x_{i(k'')}^{(1)},x_{i(k'')}^{(2)})=\emptyset$.
Hence, $R_{i(k')}$ and $R_{i(k'')}$ are disjoint. Therefore, the
rectangles $R_{i(1)},\dots,R_{i(k)}$ are pairwise disjoint and thus
form a feasible solution.

The algorithm for constructing the horizontal lines works symmetrically.
\end{proof}

We apply the algorithms due to Lemma~\ref{lem:grid}. If one of them
finds a set of $k$ independent rectangles then we output them and
we are done. Otherwise, we obtain the sets $\L_{V}$ and $\L_{H}$.
For convenience, we define two more vertical lines with $x$-coordinates $\ell_{0}^{V}:=0$ and $\ell_{|\L_{V}|+1}^{V}=2n-1$, resp.,
and similarly two more horizontal lines with $y$-coordinates $\ell_{0}^{H}=0$
and $\ell_{|\L_{H}|+1}^{H}=2n-1$, resp.. We denote by $\G$
the set of grid cells formed by these lines and the lines in $\L_{V}\cup\L_{H}$:
for any two consecutive vertices lines (i.e., defined via $x$-coordinates
$\ell_{j}^{V},\ell_{j+1}^{V}$ with $j\in\{0,\dots,|\L_{V}|\}$) and
two consecutive horizontal grid lines (defined via $y$-coordinates
$\ell_{j'}^{H},\ell_{j'+1}^{H}$ with $j'\in\{0,\dots,|\L_{H}|\}$)
we obtain a grid cell whose corners are the intersection of these
respective lines. We interpret the grid cells as closed sets (i.e.,
two adjacent grid cells intersect on their boundary).

\begin{proposition}
Each input rectangle $R_{i}$ contains a corner of a grid cell of
$\G$. If a rectangle $R$ intersects a grid cell $g$ then it must
contain a corner of $g$.
\end{proposition}

\subsection{\label{subsec:Groups-of-rectangles}Groups of rectangles}

Let $\R^{*}$ denote a solution to the given instance with $|\R^{*}|=k$.
We prove that there is a special solution $\R'\subseteq\R^{*}$ of large cardinality that
we can partition into $s\le k$ groups $\R'_{1}\dot{\cup}\dots\dot{\cup}\R'_{s}$
such that each group has constant
size $O(1/\epsilon^{8})$
and no grid cell can be intersected by rectangles from different
groups. The remainder of this section is devoted to proving the following lemma.
\begin{lemma}
\label{lem:exists-structured-solution}There is a constant $c=O(1/\epsilon^{8})$
such that there exists a solution $\R'\subseteq\R^{*}$ with $|\R'|\ge(1-\epsilon)|\R^{*}|$
and a partition $\R'=\R'_{1}\dot{\cup}\dots\dot{\cup}\R'_{s}$ 
with $s\le k$ and
$|\R'_{j}|\le c$ for each $j$ and such that if any
two rectangles in $\R'$ intersect the same grid cell $g\in\G$ then
they are contained in the same set $\R'_{j}$. 
\end{lemma}
Given the solution $\R^{*}$ we construct a planar graph $G_{1}=(V_{1},E_{1})$.
In $V_{1}$ we have one vertex $v_{i}$ for each rectangle $R_{i}\in\R^{*}$.
We connect two vertices $v_{i},v_{i'}$ by an edge if and only if
there is a grid cell $g\in\G$ such that $R_{i}$ and $R_{i'}$ intersect
$g$ and
\begin{itemize}
\item $R_{i}$ and $R_{i'}$ are crossed by the same horizontal or vertical
line in $\L_{V}\cup\L_{H}$ or if
\item $R_{i}$ and $R_{i'}$ contain the top left and the bottom right corner
of $g$, resp.
\end{itemize}
Note that we do not introduce an edge if $R_{i}$ and $R_{i'}$ contain
the bottom left and the top right corner of $g$, resp. (see
Fig.~\ref{fig:grid}): this way we preserve the planarity of the resulting graph, however we will have to deal with the missing connections in a later stage.
\begin{lemma}
\label{lem:G1-planar}The graph $G_{1}$ is planar. 
\end{lemma}

Next, we use a result by Frederickson~\cite{federickson1987fast}
to obtain a subgraph $G'_{1}$ of $G_{1}$ in which each connected
component has constant size. 

\begin{lemma}
\label{lem:apply-separator}Let $\epsilon'>0$. There exists a value
$c'=O(1/(\epsilon')^{2})$ such that the following holds: let $G=(V,E)$
be a planar graph. There exists a set of vertices $V'\subseteq V$
with $|V'|\ge(1-\epsilon')|V|$ such that in the graph $G':=G[V']$
each connected component has at most $c'$ vertices.
\end{lemma}

Let $G'_{1}$ be the graph obtained when applying Lemma~\ref{lem:apply-separator}
to $G_1$ with $\epsilon':=\epsilon/2$ and let $c_{1}=O((1/\epsilon)^{2})$
be the respective value $c'$. Now we would like to claim that if
two rectangles $R_{i},R_{i'}$ intersect the same grid cell $g\in\G$
then $v_{i},v_{i'}$ are in the same component of $G_{1}'$. Unfortunately,
this is not true. It might be that there is a grid cell $g\in\G$
such that $R_{i}$ and $R_{i'}$ contain the bottom left corner and
the top right corner of $g$, resp., and that $v_{i}$ and
$v_{i'}$ are in different components of $G_{1}'$. We fix this in
a second step. We define a graph $G_{2}=(V_{2},E_{2})$. In $V_{2}$
we have one vertex for each connected component in $G_{1}'$. We connect
two vertices $w_{i},w_{i'}\in V_{2}$ by an edge if and only if there
are two rectangles $R_{i},R_{i'}$ such that their corresponding vertices
$v_{i},v_{i'}$ in $V_{1}$ belong to the connected components of
$G_{1}'$ represented by $w_{i}$ and $w_{i'}$, resp., and
there is a grid cell $g$ whose bottom left and top right corner are
contained in $R_{i}$ and $R_{i'}$, resp. 

\begin{lemma}\label{lem:G2-planar}
The graph $G_{2}$ is planar.
\end{lemma}

Similarly as above, we apply Lemma~\ref{lem:apply-separator} to
$G_{2}$ with $\epsilon':=\frac{\epsilon}{2c_{1}}$ and let $c_{2}=O((1/\epsilon')^{2})=O(1/\epsilon^{6})$
denote the corresponding value of $c'$. Denote by $G_{2}'$ the resulting
graph.
We define a group $\R'_{q}$ for each connected component $\C_{q}$
of $V_{2}'$. The set $\R'_{q}$ contains all rectangles $R_{i}$
such that $v_{i}$ is contained in a connected component $C_{j}$
of $G_{1}'$ such that $w_{j}\in\C_{q}$. We define $\R':=\dot{\cup}_{q}\R'_{q}$.
\begin{lemma}
Let $R_{i},R_{i'}\in\R'$ be rectangles that intersect the same grid
cell $g\in\G$. Then there is a set $\R'_{q}$ such that $\{R_{i},R_{i'}\}\subseteq\R'_{q}$.
\end{lemma}
\begin{proof}
Assume that in $G_{1}$ there is an edge connecting $v_{i},v_{i'}$.
Then the latter vertices are in the same connected component $C_{j'}$
of $G_{1}'$ and thus they are in the same group $\R'_{q}$. Otherwise,
if there is no edge connecting $v_{i},v_{i'}$ in $G_{1}$ then $R_{i}$
and $R_{i'}$ contain the bottom left and top right corners of $g$,
resp. Assume that $v_{i}$ and $v_{i'}$ are contained in
the connected components $C_{j}$ and $C_{j'}$ of $G_{1}'$, resp.
Then $w_{j},w_{j'}\in V_{2}'$, $\{w_{j},w_{j'}\}\in E_{2}$ and
$w_{j},w_{j'}$ are in the same connected component of $V_{2}'$.
Hence, $R_{i},R_{i'}$ are in the same group $\R'_{q}$. 
\end{proof}
It remains to prove that each group $\R'_{q}$ has constant size
and that $|\R'|\ge(1-\epsilon)|\R^{*}|$. 
\begin{lemma}
There is a constant $c=O(1/\epsilon^{8})$ such that for each group
$\R'_{q}$ it holds that $|\R'_{q}|\le c$.
\end{lemma}
\begin{proof}
For each group $\R'_{q}$ there is a connected component $\C_{q}$
of $G_{2}'$ such that $\R'_{q}$ contains all rectangles $R_{i}$
such that $v_{i}$ is contained in a connected component $C_{j}$
of $G_{1}'$ and $w_{j}\in\C_{q}$. Each connected component of $G_{1}'$
contains at most $c_{1}=O(1/\eps^2)$ vertices of $V_{1}'$ and
each component of $G_{2}'$ contains at most $c_{2}=O(1/\eps^6)$
vertices of $V_{2}'$. Hence, $|\R'_{q}|\le c_{1}\cdot c_{2}=:c$
and $c=O((1/\epsilon^{2})(1/\epsilon^{6}))=O(1/\epsilon^{8})$.
\end{proof}
\begin{lemma}
We have that $|\R'|\ge(1-\epsilon)|\R^{*}|$. 
\end{lemma}
\begin{proof}
At most $\frac{\epsilon}{2}\cdot|V_{1}|$ vertices of $G_{1}$ are
deleted when we construct $G_{1}'$ from $G_{1}$. Each vertex in
$G_{1}'$ belongs to one connected component $C_{j}$, represented
by a vertex $w_{j}\in G_{2}$. At most $\frac{\epsilon}{2c_{1}}|V_{2}|$
vertices are deleted when we construct $G_{2}'$ from $G_{2}$. These
vertices represent at most $c_{1}\cdot\frac{\epsilon}{2c_{1}}|V_{2}|\le\frac{\epsilon}{2}|V'_{1}|\le\frac{\epsilon}{2}|V{}_{1}|$
vertices in $G_{1}$ (and each vertex in $G_{1}$ represents one rectangle
in $\R^{*}$). Therefore, $|\R'|\ge|\R^{*}|-\frac{\epsilon}{2}\cdot|V_{1}|-\frac{\epsilon}{2}\cdot|V_{1}|=(1-\epsilon)|\R^{*}|$.
\end{proof}
This completes the proof of Lemma~\ref{lem:exists-structured-solution}.

\subsection{The algorithm}
\label{sec:pas}

In our algorithm, we compute a solution that is at least as good as
the solution $\R'$ as given by Lemma~\ref{lem:exists-structured-solution}.
For each group $\R'_{j}$ we define by $\G_{j}$ the set of grid
cells that are intersected by at least one rectangle from $\R'_{j}$.
Since in $\R'$ each grid cell can be intersected by rectangles of
only one group, we have that $\G_{j}\cap\G_{q}=\emptyset$ if $j\ne q$.
We want to guess the sets $\G_{j}$. The next lemma shows that the number of possibilities for one of
those sets is polynomially bounded in $k$.
\begin{lemma}
\label{lem:shape-cells-group-1}Each $\G_{j}$ belongs to a set $\G$ of cardinality at most $k^{O(1/\eps^8)}$ that can be computed in polynomial time. 
\end{lemma}
\begin{proof}
The cells $\G_j$ intersected by $\R'_{j}$ are the union of all cells $\G(R)$ with 
$R\in \R'_j$ where for each rectangle $R$ the set $\G(R)$ denotes the cells intersected by $R$. Each set $\G(R)$ can be specified by indicating the $4$ \emph{corner} cells of $\G(R)$, i.e., top-left, top-right, bottom-left, and bottom-right corner. Hence there are at most $k^4$ choices for each such $R$. The claim follows since $|\R'_{j}|=O(1/\eps^8)$. 
\end{proof}
We hence achieve the main result of this section.

\begin{proof}[Proof of Theorem \ref{thm:MISR-FPT-PTAS}]
Using Lemma \ref{lem:shape-cells-group-1}, we can guess by exhaustive enumeration all the sets $\G_j$ in time $k^{O(k/\epsilon^{8})}$. We obtain
one independent problem for each value $j\in\{1,\dots, s\}$ 
which consists
of all input rectangles that are contained in $\G_{j}$. For this
subproblem, it suffices to compute a solution with at least $|\R'_{j}|$
rectangles. Since $|\R'_{j}|\le c=O(1/\epsilon^{8})$ we can do this
in time $n^{O(1/\epsilon^{8})}$ by complete enumeration. Thus, we solve
each of the subproblems and output the union of the computed solutions. 
The overall running time is as in the claim.
If all the computed solutions have size less than $(1-\eps)k$, this implies that the optimum solution is smaller than $k$.
Otherwise we obtain a solution of size at least $(1-\eps)k \ge k/(1+2\eps)$ and the claim follows by redefining $\eps$ appropriately.
\end{proof}

Essentially the same construction as above also gives an approximate kernelization algorithm as claimed in Theorem~\ref{thm:MISR-FPT-PTAS-kernel}, see Appendix~\ref{apx:omitted-proofs} for details.

\section{A Parameterized Approximation Scheme for 2DKR}
\label{sec:2dkr}

In this section we present a PAS and an approximate kernelization
for \TDKR. W.l.o.g., we assume that $k\ge\Omega(1/\epsilon^{3})$, since otherwise we can optimally solve the problem in time $n^{O(1/\epsilon^{3})}$ by exhaustive enumeration. In Section \ref{subsec:free-space-solution} we show that, if a solution of size $k$ exists,  
there is a solution of size at least $(1-\epsilon)k$ in which no item intersects
some horizontal strip $(0,N)\times (0,(1/k)^{O(1/\epsilon)}N)$ at the bottom of the knapsack.
In Section \ref{subsec:FPT-RA} we show that, if there exists a solution of size $k'$ that does not use the mentioned strip, then we can compute in polynomial time a set of size $(k')^{O(1/\epsilon)}$ that contains a solution of size $k'$ (where we are allowed to use the full knapsack).  Combining these two results gives 
Theorem~\ref{thr:tdkr:pas}.

\subsection{\label{subsec:free-space-solution}Freeing a Horizontal Strip}

In this section, we prove the following lemma that shows the existence
of a near-optimal solution that leaves a sufficiently tall empty horizontal strip in the knapsack
(assuming $k\geq \Omega(1/\epsilon^{3})$). W.l.o.g., $\eps\leq 1$. Since
we can rotate the items by 90 degrees, we can assume w.l.o.g.~that
$w_{i}\ge h_{i}$ for each item $i\in I$.

\begin{lemma}
\label{lem:free-space}Let $k\in\mathbb{N}$, $k = \Omega(1/\epsilon^{3})$, and $\epsilon>0$. Given
an instance of \TDKR with a solution of size $k$,
there exists a solution of size at least $(1-\epsilon)k$ in which
no packed item intersects $(0,N)\times (0,(1/k)^{c}N)$, for a proper constant $c=O(1/\epsilon)$.
\end{lemma}
We classify items into large and thin items. Via a shifting argument,
we get the following lemma. 
\begin{lemma}
\label{lem:shifting}There is an integer $\sp\in\{1,\dots,\lceil8/\epsilon\rceil\}$
such that by losing a factor of $1+\epsilon$ in the objective we
can assume that the input items are partitioned into 
\begin{itemize}
\item \emph{large} items $L$ such that $h_{i}\ge(1/k)^{\sp}N$ (and thus also
$w_{i}\ge(1/k)^{\sp}N$) for each item $i\in L$, 
\item \emph{thin} items $T$ such that $h_{i}<(1/k)^{\sp+2}N$ for each item $i\in T$.
\end{itemize}
\end{lemma}

Let $B$ be the integer due to Lemma~\ref{lem:shifting} and we work
with the resulting item classification. If $|T|\ge k$ then we can
create a solution of size $k$ satisfying the claim of Lemma~\ref{lem:free-space}
by simply stacking $k$ thin items on top of each other: any $k$
thin items have a total height of at most $k\cdot(1/k)^{\sp+2}N\le(1/k)^{2}N$.
Thus, from now on assume that $|T|< k$. 

\subparagraph{Sparsifying large items.}

Our strategy is now to delete some of the large items and move the
remaining items. This will allow us to free the area $[0,N]\times[0,(1/k)^{O(1/\epsilon)}N]$
of the knapsack. Denote by $\OPT'$ the almost optimal solution obtained by applying
Lemma~\ref{lem:shifting}. We remove the items in $\OPT'_T:=\OPT'\cap T$ temporarily;
we will add them back later. 

We construct a directed graph $G=(V,A)$ where we have one vertex
$v_{i}\in V$ for each item $i\in \OPT'_L:=\OPT'\cap L$. We connect two vertices
$v_{i},v_{i'}$ by an arc $a=(v_{i},v_{i'})$ if and only if we can
draw a vertical line segment of length at most $(1/k)^{\sp}N$ that
connects item $i$ with item $i'$ without intersecting any other
item such that $i'$ lies above $i$, i.e., the 
bottom coordinate of $i'$ is at least as large as the top
coordinate of $i$, see Figure~\ref{fig:2DKP} for
a sketch.
We obtain the following proposition since 
for each edge we can draw a vertical line segment and these segments do
not intersect each other.

\begin{figure}[t]
\begin{centering}
\includegraphics[scale=0.5]{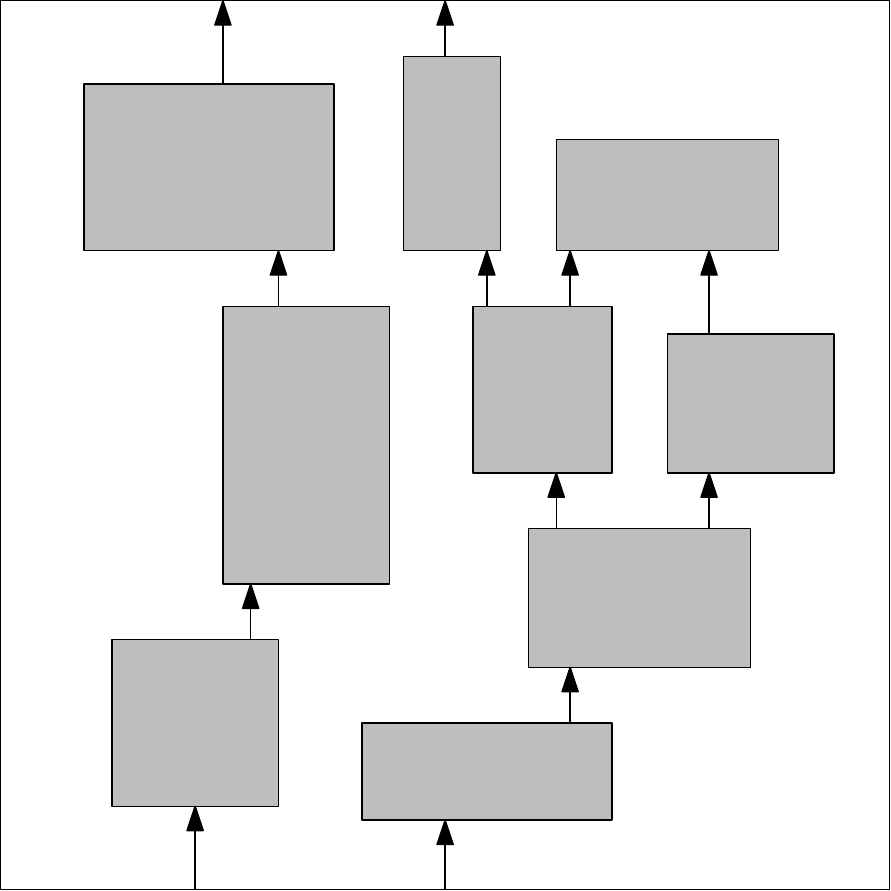}~~~~~~~~~~~~~\includegraphics[scale=0.5]{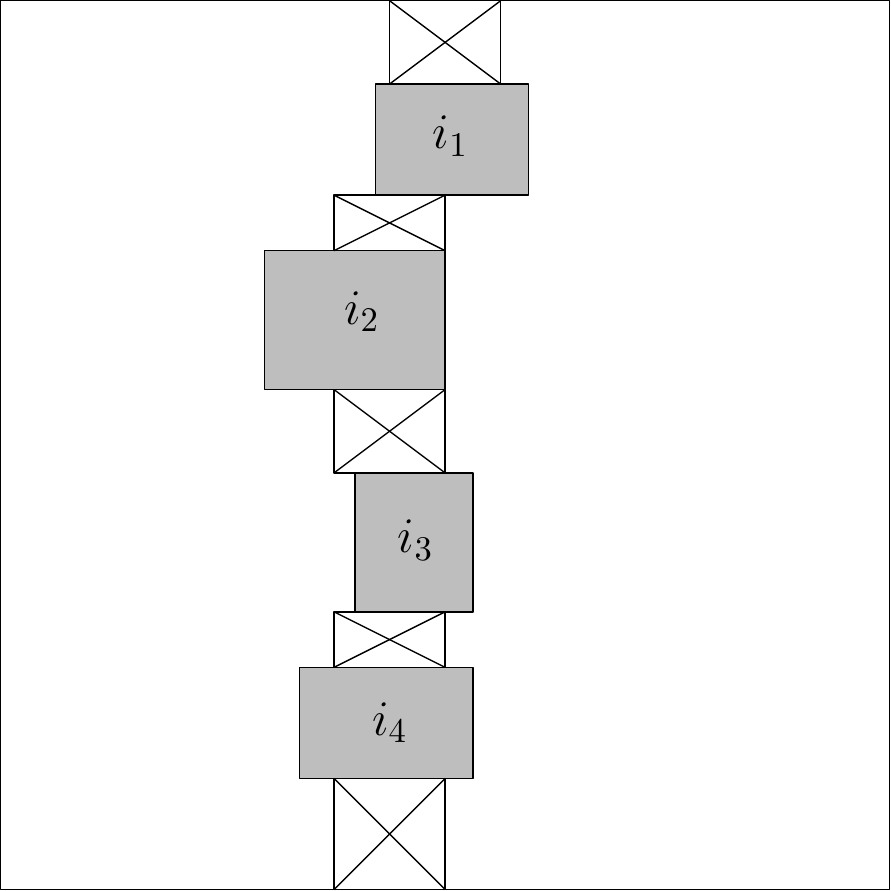} 
\par\end{centering}
\caption{\label{fig:2DKP}The left figure shows the arcs of the graph $G$.
Each item corresponds to one vertex of the graph. The right figure
shows the items $i_{1},\dots,i_{K}$ and the deletion rectangles between
them.}
\end{figure}

\begin{proposition} The graph $G$ is planar. \end{proposition} 

Next, we apply Lemma~\ref{lem:apply-separator} to $G$ with $\epsilon':=\epsilon$.
Let $G'=(V',A')$ be the resulting graph. We remove from $\OPT'_L$ all
items $i\in V\setminus V'$ and denote by $\OPT''_L$ the resulting solution.
We push up all items in $\OPT''_L$
as much as possible. If now the strip $(0,N)\times (0,(1/k)^{\sp}N)$ is not intersected by any item then
we can place all the items in $T$ into the remaining
space. Their total height can be at most $k\cdot(1/k)^{\sp+2}N\le(1/k)^{\sp+1}N$
and thus we can leave a strip of height $(1/k)^{\sp}N-(1/k)^{\sp+1}N\ge(1/k)^{O(1/\epsilon)}N$
and width $N$ empty. This completes the proof of Lemma~\ref{lem:free-space} for this case. 

Assume next that the strip $(0,N)\times (0,(1/k)^{\sp}N)$ is intersected by some
item: the following lemma implies that there is a set of $c'=O(1/\epsilon^{2})$
vertices whose items intuitively connect the top and the bottom edge
of the knapsack.
\begin{lemma} \label{lem:separating-path}Assume that in $\OPT''_L$ there
is an item $i_{1}$ intersecting $(0,N)\times (0,(1/k)^{\sp}N)$. Then $G$ contains a path $v_{i_{1}},v_{i_{2}},\dots,v_{i_{K}}$ with
$K\le c'=O(1/\epsilon^{2})$, such that the distance between $i_{K}$ and the top edge of the knapsack
is less than $(1/k)^{\sp}N$.
\end{lemma} 
\begin{proof} 
Let $C$ denote all vertices $v$
in $G'$ such that there is a directed path from $v_{i_{1}}$ to $v$
in $G'$. The vertices in $C$ are contained in the connected component
$C'$ in $G'$ that contains $v_{i_{1}}$. Note that $|C|\le|C'|\le c'$.
We claim that $C$ must contain a vertex $v_{j}$ whose corresponding
item $j$ is closer than $(1/k)^{\sp}N$ to the top edge of the knapsack.
Otherwise, we would have been able to push up all items corresponding
to vertices in $C$ by $(1/k)^{\sp}N$ units: first we could have pushed
up all items such that their corresponding vertices have no outgoing
arc, then 
all items such that their vertices have outgoing
arcs pointing at the former set of vertices, and so on. 
By definition of
$C$, there must be a path connecting $v_{i_{1}}$ with $v_{j}$.
This path $v_{i_{1}},v_{i_{2}},\dots,v_{i_{K}}=v_{j}$ contains only vertices in $C$ and hence its length is bounded
by $c'$. The claim follows. \end{proof} 

Our goal is now to remove the items $i_{1},\dots,i_{K}$ due to Lemma~\ref{lem:separating-path}
and $O(K)=O(1/\epsilon^{2})$ more large items from $\OPT''_L$. Since
we can assume that $k\ge\Omega(1/\epsilon^{3})$ this will lose only
a factor of $1+O(\epsilon)$ in the objective.
To this end we define $K+1$ \emph{deletion rectangles}, see Figure~\ref{fig:2DKP}.
We place one such rectangle $R_{\ell}$ between any two consecutive
items $i_{\ell},i_{\ell+1}$. The height of $R_{\ell}$ equals the
vertical distance between $i_{\ell}$ and $i_{\ell+1}$ (at most $(1/k)^{\sp}N$)
and the width of $R_{\ell}$ equals $(1/k)^{\sp}N$. Since $v_{i_{\ell}},v_{i_{\ell+1}}$
are connected by an arc in $G'$, we can draw a vertical line segment
connecting $i_{\ell}$ with $i_{\ell+1}$. We place $R_{\ell}$ such
that it is intersected by this line segment. Note that for the horizontal
position of $R_{\ell}$ there are still several possibilities and
we choose one arbitrarily.
Finally, we place a special deletion rectangle between the item $i_{K}$
and the top edge of the knapsack and another special deletion rectangle
between the item $i_{1}$ and the bottom edge of the knapsack. The
heights of these rectangles equal the distance of $i_{1}$ and $i_{K}$
with the bottom and top edge of the knapsack, resp. (which is at most
$(1/k)^{B}N$), and their widths equal $(1/k)^{\sp}N$. They are placed
such that they touch the bottom edge of $i_{1}$ and the top edge
of $i_{K}$, resp. 

\begin{lemma}\label{lem:deletion-rectangles} Each deletion rectangle can
intersect at most $4$ large items in its interior. Hence, there
can be only $O(K)\le O(c')=O(1/\epsilon^{2})$ large items intersecting
a deletion rectangle in their interior. \end{lemma} 

Observe that the deletion rectangles and the items in $\{i_{1},\dots,i_{K}\}$
separate the knapsack into a left and a right part with items $\OPT''_{left}$
and $\OPT''_{right}$, resp. We delete all items in $i_{1},\dots,i_{K}$
and all items intersecting the interior of a deletion rectangle. Each
deletion rectangle and each item in $\{i_{1},\dots,i_{K}\}$ has a width
of at least $(1/k)^{\sp}N$. Thus, we can move all items in $\OPT''_{left}$
simultaneously by $(1/k)^{\sp}N$ units to the right. After this,
no large item intersects the area $(0,(1/k)^{\sp}N)\times (0,N)$.
We rotate the resulting solution by 90 degrees, hence getting an empty horizontal strip  
$(0,N)\times (0,(1/k)^{\sp}N)$. The total height of items in $OPT'_T$ is at most $k\cdot(1/k)^{\sp+2}N\le(1/k)^{\sp+1}N$. Therefore, the items in $OPT'_T$ can be stacked (one on top of the other) inside a horizontal strip of height $(1/k)^{\sp+1}N$ that can be placed right below the rectangles in $\OPT''_{left}\cup \OPT''_{right}$. This leaves an empty horizontal strip of height
$(1/k)^{\sp}N-(1/k)^{\sp+1}N\ge(1/k)^{O(1/\epsilon)}N$ at the bottom of the knapsack.  
This completes the proof of Lemma~\ref{lem:free-space}. 

\subsection{FPT-algorithm with resource augmentation\label{subsec:FPT-RA}}

We 
now compute a packing that contains as many items as the solution
due to Lemma~\ref{lem:free-space}. However, it might use the space
of the entire knapsack. In particular, we use the free space in the
knapsack in the latter solution in order to round the sizes of the
items. In the following lemma the reader may think of $k'=(1-\epsilon)k$
and $\tilde{k}=k^{O(1/\epsilon)}$.

\begin{lemma}
\label{lem:resource-augmentation}Let $k',\tilde{k}\in\mathbb{N}$.
There is an algorithm for \TDKR with a running time of $(\tilde{k}k')^{O(k')}n^{O(1)}$ 
that computes
a solution of size $k'$ or asserts that there is no solution of size $k'$ fitting
into a restricted knapsack $[0,N]\times [0,(1-1/\tilde{k})N]$. Also,
in time $n^{O(1)}$ we can compute a set of size $O(\tilde{k}(k')^{2})$ that
contains a solution of size $k'$ if there is such a solution that
fits into the latter knapsack.
\end{lemma}
Note that Lemma~\ref{lem:resource-augmentation} yields an FPT 
algorithm if we are allowed to increase the size of the knapsack by
a factor $1+O(1/\tilde{k})$ where $\tilde{k}$ is a second parameter.

In the remainder of this section, we prove Lemma~\ref{lem:resource-augmentation}
and we do not differentiate between large and thin items anymore.
Assume that there exists a solution $\OPT''$ of size $k'$ that leaves
the area $[0,N]\times[0,N/\tilde{k}]$ of the knapsack empty. We want
to compute a solution of size $k'$. We use the empty space 
in order to round the heights of the items in the packing of $\OPT''$ to integral
multiples of $N/(k'\tilde{k})$. Note that in $\OPT''$ an item $i$
might be rotated. Thus, depending on this we actually want to round
its height $h_{i}$ or its width $w_{i}$. To this end, we define
rounded heights and widths by $\hat{h}_{i}:=\left\lceil \frac{h_{i}}{N/(k'\tilde{k})}\right\rceil N/(k'\tilde{k})$
and $\hat{w}_{i}:=\left\lceil \frac{h_{i}}{N/(k'\tilde{k})}\right\rceil N/(k'\tilde{k})$
for each item $i$. 
\begin{lemma} \label{lem:rounded-packing}There
exists a feasible packing for all items in $\OPT''$ even if for each
rotated item $i$ we increase its width $w_{i}$ to $\hat{w}_{i}$
and for each non-rotated item $i'\in\OPT''$ we increase its height
$h_{i'}$ to $\hat{h}_{i'}$. \end{lemma} 

To visualize the packing due to Lemma~\ref{lem:rounded-packing}
one might imagine a container of height $\hat{h}_{i}$ and width $w_{i}$
for each non-rotated item $i$ and a container of height $h_{i'}$
and width $\hat{w}_{i'}$ for each rotated item $i'$. Next, we group
the items according to their values $\hat{h}_{i}$ and $\hat{w}_{i}$.
We define $I_{h}^{(j)}:=\{i\in I\mid\hat{h}_{i}=jN/(k'\tilde{k})\}$
and $I_{w}^{(j)}:=\{i\in I\mid\hat{w}_{i}=jN/(k'\tilde{k})\}$ for
each $j\in\{1,\dots,k'\tilde{k}\}$. The crucial observation is now
that from each set $I_{h}^{(j)}$ it suffices to consider only the
$k'$ items with smallest width. If $\OPT''$ uses an item from $I_{h}^{(j)}$
with larger width then we can replace it by one of the $k'$ thinner items that
is not contained in $\OPT''$. A symmetric statement holds for the
sets $I_{v}^{(j)}$. 

\begin{lemma} \label{lem:k-items-per-set}We can assume that from each set
$I_{h}^{(j)}$ the solution $\OPT''$ contains only items among the
$k'$ items in $I_{h}^{(j)}$ with smallest width. Similarly, from
each set $I_{w}^{(j)}$ the solution $\OPT''$ contains only items
among the $k'$ items in $I_{w}^{(j)}$ with smallest height. 
\end{lemma}

We eliminate from each set $L_{h}^{(j)}$ and $L_{w}^{(j)}$ the items
that are not among the $k'$ items with smallest width and height,
resp. At most $2k'\cdot k'\tilde{k}=O(\tilde{k}(k')^{2})$
items remain, denote them by $\bar{I}$. 
Then, in time $(\tilde{k}k')^{O(k')}$ we can solve the remaining
problem by completely enumerating over all subsets of $\bar{I}$ with
at most $k'$ elements. For each enumerated set we 
check 
within the given time bounds 
whether its items can be packed into the knapsack (possibly
via rotating some of them)
by guessing sufficient auxiliary information.
Therefore, if a solution of size $k'$ for a knapsack
of width $N$ and height $(1-1/\tilde{k})N$ exists, then we will
find a solution of size $k'$ that fits into a knapsack of width and
height $N$.

Now the proof of Theorem \ref{thr:tdkr:pas} 
follows by using
Lemma~\ref{lem:free-space} and then applying Lemma~\ref{lem:resource-augmentation}
with $k'=(1-\epsilon)k$ and $\tilde{k}=k^{O(1/\epsilon)}$. The set $\bar{I}$ is the claimed set (which intuitively forms the approximative kernel), 
we compute a solution of size at least $(1-\eps)k\ge k/(1+2\eps)$
and we can redefine $\eps$ appropriately.

\section{Hardness of Geometric Knapsack}\label{section:hardness}

We show that \TDK and \TDKR are both \W{1}-hard for parameter $k$ by reducing from a variant of \SSum. Recall that in \SSum we are given $m$ positive integers $x_1,\ldots,x_{m}$ as well as integers $t$ and $k$, and have to determine whether some $k$-tuple of the numbers sums to $t$; this is \W{1}-hard with respect to $k$~\cite{DowneyF95}. In the variant \MSS it is allowed to choose numbers more than once. It is easy to verify that the proof for \W{1}-hardness of \SSum due to Downey and Fellows~\cite{DowneyF95} extends also to \MSS. 
(See Lemma~\ref{lemma:mssk:hardness} in Section~\ref{appendix:hardness:proofs}.)
In our reduction to \TDKR we prove that rotations are not required for optimal solutions, making \W{1}-hardness of \TDK a free consequence.

\begin{proof}[Proof sketch for Theorem~\ref{thm:w1hard}.]
We give a polynomial-time parameterized reduction from \MSS to \TDKR with output parameter $k'=O(k^2)$. 
% By Lemma~\ref{lemma:mssk:hardness}, this establishes \W{1}-hardness of \TDKR.
This establishes \W{1}-hardness of \TDKR.

Observe that, for any packing of items into the knapsack, there is an upper bound of $N$ on the total width of items that intersect any horizontal line through the knapsack, and similarly an upper bound of $N$ for the total height of items along any vertical line. We will let the dimensions of some items depend on numbers $x_i$ from the input instance $(x_1,\ldots,x_{m},t,k)$ of \MSS such that, using these upper bound inequalities, a correct packing certifies that $y_1+\ldots+y_k=t$ for some $k$ of the numbers. The key difficulty is that there is a lot of freedom in the choice of which items to pack and where in case of a no instance.

To deal with this, the items corresponding to numbers $x_i$ from the input are all \emph{almost} squares and their dimensions are incomparable. Concretely, an item corresponding to some number $x_i$ has height $L+S+x_i$ and width $L+S+2t-x_i$; we call such an item a \emph{tile}. (The exact values of $L$ and $S$ are immaterial here, but $L \gg S \gg t > x_i$ holds.) Thus, when using, e.g., a tile of smaller width (i.e., smaller value of $x_i$) it will occupy ``more height'' in the packing. The knapsack is only slightly larger than a $k$ by $k$ grid of such tiles, implying that there is little freedom for the placement. Let us also assume for the moment, that no rotations are used.

Accordingly, we can specify $k$ vertical lines that are guaranteed to intersect all tiles of any packing that uses $k^2$ tiles, by using pairwise distance $L-1$ between them. Moreover, each line is intersecting exactly $k$ private tiles. The same holds for a similar set of $k$ horizontal lines. Together we get an upper bound of $N$ for the sum of the widths (heights) along any horizontal (vertical) line. Since the numbers $x_i$ occur negatively in widths, we effectively get lower bounds for them from the horizontal lines. When the sizes of these tiles (and the auxiliary items below) are appropriately chosen, it follows that all upper bound equalities must be tight. This in turn, due to the exact choice of $N$, implies that there are $k$ numbers $y_1,\ldots,y_k$ with sum equal to $t$. 

Unsurprisingly, using just the tiles we cannot guarantee that a packing exists when given a yes-instance. This can be fixed by adding a small number of flat/thin items that can be inserted between the tiles (see Figure~\ref{figure:hardness}, but note that it does not match the size ratios from this proof); these have dimension $L\times S$ or $S\times L$. Because one dimension of these items is large (namely $L$) they must be intersected by the above horizontal or vertical lines. Thus, they can be proved to enter the above inequalities in a uniform way, so that the proof idea goes through.

Finally, let us address the question of why we can assume that there are no rotations. This is achieved by letting the width of any tile be larger than the height of any tile, and adding a final auxiliary item of width $N$ and small height, called the \emph{bar}. To get the desired number of items in a solution packing, it can be ensured that the bar must be used as no more than $k^2$ tiles can fit into $N\times N$ and there is a limited supply of flat/thin items. W.l.o.g., the bar is not rotated. It can then be checked that using at least one tile in its rotated form will violate one of the upper bounds for the height. This completes the proof sketch.
\end{proof}

\section{Open Problems}

This paper leaves several interesting open problems. A first obvious question is whether there exists a PAS also for $\TDK$ (i.e., in the case without rotations). We remark that the algorithm from Lemma~\ref{lem:resource-augmentation}
can be easily adapted to the case without rotations. 
Unfortunately, Lemma~\ref{lem:free-space} does not seem to generalize to the latter case. Indeed, there are instances in which we lose up to a factor of $2$
if we require a strip of width $\Omega_{\eps,k}(1)\cdot N$ to be emptied, see Figure~\ref{figure:counter-example}.
We also note that both our PASs work for the cardinality version of the problems: an extension to the weighted case is desirable. Unlike related results in the literature (where extension to the weighted case follows relatively easily from the cardinality case), this seems to pose several technical issues.

We remark that all the problems considered in this paper might admit a PTAS in the standard sense, which would be a strict improvement on our PASs. Indeed, the existence of a QPTAS for these problems \cite{adamaszek2013approximation,Adamaszek2015,chuzhoy2016approximating} suggests that such PTASs are likely to exist. However, finding those PTASs is a very well-known and long-standing problem in the area. We hope that our results can help to achieve this challenging goal.

%%
%% Bibliography
%%

%% Please use bibtex, 

\bibliography{references}

\begin{figure}[p]
\centering
\begin{tikzpicture}[scale=0.04,thick]
\draw (0,0) rectangle (130,130);

\draw (0,0) rectangle (39,31);
\draw (5,5) node {$1$};
\draw[fill=black!30!white] (39,5) rectangle (49,30);
\draw (49,0) rectangle (86,33);
\draw (54,5) node {$3$};
\draw[fill=black!30!white] (86,5) rectangle (96,30);
\draw (96,0) rectangle (130,36);
\draw (101,5) node {$6$};

\draw[fill=black!30!white] (5,31) rectangle (30,41);
\draw[fill=black!30!white] (52,33) rectangle (77,43);
\draw[fill=black!30!white] (100,36) rectangle (125,46);

\draw (0,41) rectangle (34,77);
\draw (5,46) node {$6$};
\draw[fill=black!30!white] (34,48) rectangle (44,73);
\draw (44,43) rectangle (83,74);
\draw (49,48) node {$1$};
\draw[fill=black!30!white] (83,48) rectangle (93,73);
\draw (93,46) rectangle (130,79);
\draw (98,51) node {$3$};

\draw[fill=black!30!white] (5,77) rectangle (30,87);
\draw[fill=black!30!white] (52,74) rectangle (77,84);
\draw[fill=black!30!white] (100,79) rectangle (125,89);

\draw (0,87) rectangle (37,120);
\draw (5,92) node {$3$};
\draw[fill=black!30!white] (37,90) rectangle (47,115);
\draw (47,84) rectangle (81,120);
\draw (52,89) node {$6$};
\draw[fill=black!30!white] (81,90) rectangle (91,115);
\draw (91,89) rectangle (130,120);
\draw (96,94) node {$1$};

\draw[fill=black!30!white] (0,120) rectangle (130,130);
\end{tikzpicture}
\caption{\label{figure:hardness} A sketch of the packing used in Theorem~\ref{thm:w1hard} for a solution with $k=3$ and $1+3+6=10$. Items corresponding to the same number have the same size. The figure is not to scale: The gray items should be much flatter and the clear ones should look like squares of almost identical size.}
\end{figure}
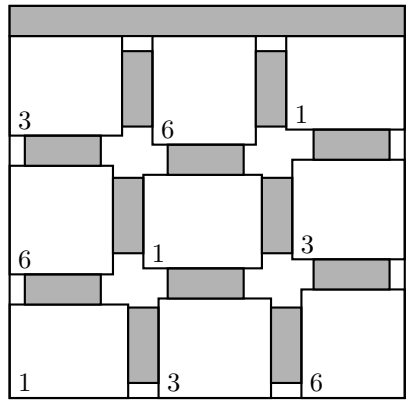

\begin{figure}[t]
\center\includegraphics[scale=0.8]{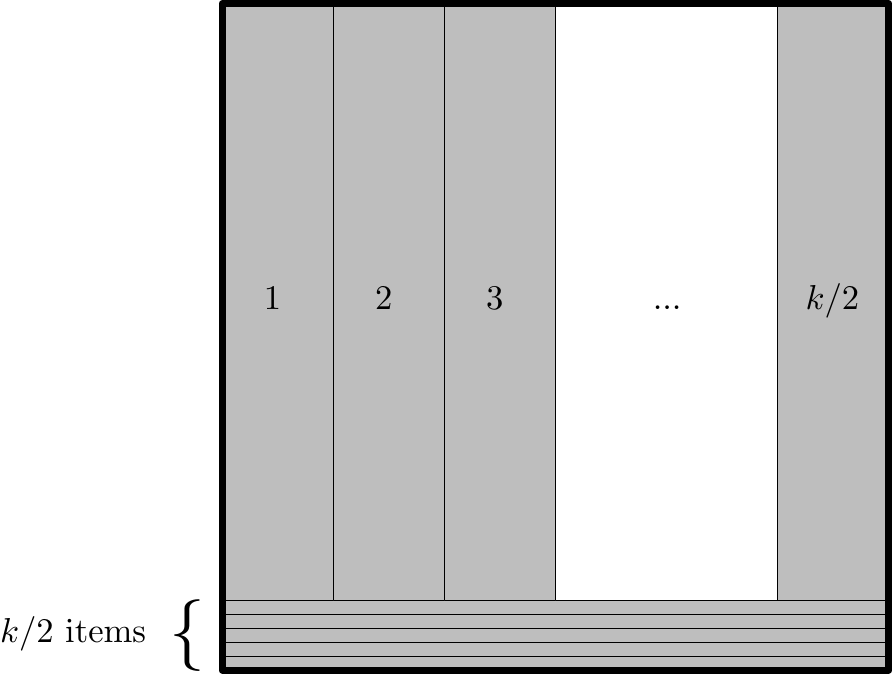}
\caption{\label{figure:counter-example}Example showing that Lemma~\ref{lem:free-space} cannot be generalized to \TDK (without rotations).
The total height of the $k/2$ items on the bottom of the knapsack can be made arbitrarily small. Suppose that we wanted to free 
up an area of height $f(k)\cdot N$ and width $N$ or of height $N$ and width $f(k)\cdot N$ (for some fixed function $f$). 
If the total height of the items on the bottom is smaller than $f(k)\cdot N$ then 
we would have to eliminate the $k/2$ items on the bottom or the $k/2$ items on top. Thus, we would lose a factor of $2>1+\eps$ in the
approximation ratio.} 
\end{figure}

\newpage
\appendix
\section{Omitted Proofs for Sections \ref{sec:misr} and \ref{sec:2dkr}\label{apx:omitted-proofs}}

\begin{proof}[Proof of Lemma~\ref{lem:G1-planar}]
We define a planar embedding for $G_{1}$ based on the position of
the rectangles in $\R^{*}$. Each vertex $v_{i}\in V_{1}$ is represented
by a rectangle $\bar{R}_{i}$ which is defined to be the convex hull
of all corners of cells of $\G$ that are contained in $R_{i}$. Let
$e=\{v_{i},v_{i'}\}\in E_{1}$ be an edge. Let $g$ be a grid cell
that $R_{i}$ and $R_{i'}$ both intersect. If $R_{i}$ and $R_{i'}$
intersect the same horizontal line $\ell^{H}\in\L_{H}$ then we represent
$e$ by a horizontal line segment $\ell'$ connecting $\bar{R}_{i}$
and $\bar{R}_{i'}$ such that $\ell^{H}$ contains $\ell'$. We do
a symmetric operation if $R_{i}$ and $R_{i'}$ intersect the same
vertical line $\ell^{V}\in\L_{V}$. If $R_{i}$ and $R_{i'}$ contain
the top left and the bottom right corner of $g$, resp., then
we represent $e$ by a diagonal line segment $\ell'$ connecting $\bar{R}_{i}$
and $\bar{R}_{i'}$ within $g$. We do this operation with each edge
$e\in E_{1}$. Note that in each grid cell we draw at most one diagonal
line segment. By construction, no two line segments intersect and
hence $G_{1}$ is planar.
\end{proof}

\begin{proof}[Proof of Lemma~\ref{lem:apply-separator}]
A result by Frederickson~\cite{federickson1987fast} states that for any integer $r$
any $n$-vertex planar graph can be divided into $O(n/r)$ regions
with no more than $r$ vertices each, and $O(n/\sqrt{r})$ boundary
vertices in total. We choose $r:=O(1/(\epsilon')^{2})$ and then we
have at most $\epsilon'\cdot n$ boundary vertices in total. We define
$V'$ to be the set of non-boundary vertices.
\end{proof}

\begin{proof}[Proof of Lemma~\ref{lem:G2-planar}]
We define a planar embedding for $G_{2}$. Let $w_{j}\in V_{2}$ and
assume that $w_{j}$ represents a connected component $C_{j}$ of
$G_{1}'$. We represent $C_{j}$ by drawing the rectangle $\bar{R}_{i}$
for each vertex $v_{i}\in C$ (like in the proof of Lemma~\ref{lem:G1-planar}
the rectangle $\bar{R}_{i}$ is defined to be the convex hull of all
corners of cells of $\G$ that are contained in $R_{i}$) and the
following set of line segments (actually almost the same as the ones
defined in the proof of Lemma~\ref{lem:G1-planar}).  Consider two
rectangles $R_{i},R_{i'}\in C_{j}$ intersecting the same grid cell
$g$. 
\begin{itemize}
\item If $R_{i},R_{i'}$ intersect the same horizontal line $\ell^{H}\in\L_{H}$
then then we draw a horizontal line segment $\ell'$ connecting $\bar{R}_{i}$
and $\bar{R}_{i'}$ such that $\ell'$ is a subset of $\ell^{H}$. 
\item If $R_{i}$ and $R_{i'}$ contain the top left and the bottom right
corner of $g$, resp., then we draw a diagonal line segment
$\ell'$ connecting $\bar{R}_{i}$ and $\bar{R}_{i'}$ within $g$. 
\end{itemize}
This yields a connected area $A_{j}$ representing $C_{j}$ (and thus
$w_{j}$). 

Let $e=\{w_{j},w_{j'}\}\in E_{2}$. We want to introduce a line segment
representing $e$. By definition of $E_{2}$ there must be grid cell
$g$ and two rectangles $R_{i}$, $R_{i'}$ intersecting $g$ whose
vertices belong to different connected components of $G_{1}'$ and
that $R_{i}$ and $R_{i'}$ contain the bottom left and the top right
corner of $g$, resp. Note that then there can be no vertex
$v_{i''}\in V_{1}'$ whose rectangle contains the top left or the
bottom right corner of $g$: such a rectangle would be connected by
an edge with both $R_{i}$ and $R_{i'}$ in $G_{1}$ and then all
three rectangles $R_{i'},R_{i'},R_{i''}$ would be in the same connected
component of $G_{1}'$. We draw a diagonal line segment $\ell'$ connecting
$\bar{R}_{i}$ and $\bar{R}_{i'}$ within $g$ and then $\ell'$ does
not intersect any area $A_{j}$ for any vertex $w_{j}\in V_{2}$.
Also, since we add at most one line segment $\ell'$ per grid cell
$g$ these line segments do not intersect each other. Hence, $G_{2}$
is planar.
\end{proof}

\begin{proof}[Proof of Theorem \ref{thm:MISR-FPT-PTAS-kernel}]
First, we define the grid as described in Section~\ref{subsec:grid}.
In case that the algorithm in Lemma~\ref{lem:grid} finds a solution
of size $k$ then we define the kernel $\bar{\R}$ to be this solution and we
are done. Otherwise, we enumerate all possible sets $\G_{k'}$ of
the kind as described in Lemma~\ref{lem:shape-cells-group-1}, at most 
$k^{O(1/\epsilon^{8})}$ many. Then,
for each such set $\G_{j}$ we consider all rectangles contained
in the union of $\G_{j}$ and we compute a feasible solution of size
$c$ for them if such a solution exists, and otherwise we compute
the optimal solution. We do this by complete enumeration in time $n^{O(c)}=n^{O(1/\epsilon^{8})}$.
For each set $\G_{j}$ the obtained solution has size at most $c=O(1/\epsilon^{8})$.
We define the kernel $\bar{\R}$ to be the union over all $k^{O(1/\epsilon^{8})}$
solutions obtained in this way. Hence, $|\bar{\R}|\le k^{O(1/\epsilon^{8})}$.
Also, we can guarantee that the output of our algorithm is a subset of 
$\bar{\R}$ and hence $\bar{\R}$ contains a $(1+\epsilon)$-approximative solution.
\end{proof}

\begin{proof}[Proof of Lemma~\ref{lem:shifting}]
Let $\OPT$ denote the optimal solution to the given instance. For
each $\sp'\in\{1,\dots,\lceil8/\epsilon\rceil\}$ we define $I(B'):=\{i\in I\mid h_{i}\in[(1/k)^{\sp'+2}N,(1/k)^{\sp'}N)\}$.
For any item $i\in I$ there can be at most four values of $B'$ such
that $i$ is contained in the respective set $I(B')$. Hence, there
must be one value $B\in\{1,\dots,\lceil8/\epsilon\rceil\}$ such that
$|I(B)\cap\OPT|\le\frac{\epsilon}{2}|\OPT|$. Each item $i\in I\setminus I(B)$
is then contained in $L$ or $T$. Since $|I(B)\cap\OPT|\le\frac{\epsilon}{2}|\OPT|$
we lose only a factor of $(1-\frac{\epsilon}{2})^{-1}\le1+\epsilon$
in the approximation ratio.
\end{proof}

\begin{proof}[Proof of Lemma~\ref{lem:deletion-rectangles}]
Each deletion rectangle has a height of at most $(1/k)^{\sp}N$ and
a width of exactly $(1/k)^{\sp}N$. Each large item has height and
width at least $(1/k)^{\sp}N$. Therefore, each deletion rectangle
can intersect with at most $4$ large items in its interior (intuitively,
at its $4$ corners).
\end{proof}

\begin{proof}[Proof of Lemma~\ref{lem:rounded-packing}] 
For each item
$i\in\OPT''$ we perform the following operation. Each item $i'\in\OPT''$
such that $i'$ is placed underneath $i$ (i.e., such that the $y$-coordinate
of the top edge of $i'$ is upper-bounded by the $y$-coordinate of
the bottom edge of $i$) is moved by $N/(k'\tilde{k})$ units down.
If $i$ is not rotated then we increase the height of $i$ to $\hat{h_{i}}$
by appending a rectangle of width $w_{i}$ and height $\hat{h}_{i}-h_{i}\le N/(k'\tilde{k})$
underneath $i$. If $i$ is rotated then we increase the width of
$i$ to $\hat{w_{i}}$ by appending a rectangle of width $h_{i}$
and height $\hat{w}_{i}-w_{i}\le N/(k'\tilde{k})$ underneath $i$.
Since we moved down the mentioned other items before, the new (bigger)
item does not intersect any other item. We do this operation for each
item $i\in\OPT''$. In the process, we move each item down by at most
$(k'-1)N/(k'\tilde{k})$ and when we increase its height then the
$y$-coordinate of its bottom edge decreases by at most $N/(k'\tilde{k})$.
Initially, the $y$-coordinate of the bottom edge of any item was
at least $N/\tilde{k}$. Hence, at the end the $y$-coordinate of
the bottom edge of any item is at least $N/\tilde{k}-(k'-1)N/(k'\tilde{k})-N/(k'\tilde{k})\ge N/\tilde{k}-N/\tilde{k}\ge0$.
Hence, all rounded items are contained in the knapsack.
\end{proof}

\begin{proof}[Proof of Lemma~\ref{lem:k-items-per-set}] 
Consider the packing for $\OPT''$ due to Lemma~\ref{lem:rounded-packing}
in which we increased the height of each non-rotated item $i$ to
$\hat{h}_{i}$ and the width of each rotated item $i'$ to $\hat{w}_{i'}$.
Suppose that there is a set $I_{h}^{(j)}$ such that $\OPT''$ contains
an item $i\in I_{h}^{(j)}$ which is not among the $k'$ items in
$I_{h}^{(j)}$ with smallest width. Denote by $\bar{I}_{h}^{(j)}$
the latter set of items. Since $|\OPT''|\le k'$ and $i\in\OPT''$
there must be an item $i'\in\bar{I}_{h}^{(j)}$ such that $i'\notin\OPT''$.
Then we can replace $i$ by $i'$ since $\hat{h}_{i'}=\hat{h}_{i}$
and $w_{i'}\le w_{i}$. We perform this operation for each set $L_{h}^{(j)}$
and a symmetric operation for each set $L_{w}^{(j)}$ until we obtain
a solution for which the lemma holds. This solution then contains
the same number of items as the initial solution $\OPT''$. 
\end{proof}

\section{Proofs for Section~\ref{section:hardness}}\label{appendix:hardness:proofs}

\begin{lemma}\label{lemma:mssk:hardness}
\MSS is \W{1}-hard.
\end{lemma}

\begin{proof}
Downey and Fellows~\cite{DowneyF95} give a parameterized reduction from \problem{Perfect Code($k$)} to \SSum. The created instances $(x_1,\ldots,x_m,t,k)$ have the property that all numbers have digits $0$ or $1$ when expressed in base $k+1$. Moreover, the target value $t$ is equal to $1\ldots 1_{k+1}$. Accordingly, when any $k$ numbers $x_i$ sum to $t$ there can be no carries in the addition. Thus, no two selected numbers may have a $1$ in the same position. Hence, allowing to select numbers multiple times does not create spurious solutions, giving us a correct reduction from \problem{Perfect Code($k$)} to \MSS.
\end{proof}

We split the proof of Theorem~\ref{thm:w1hard} into two separate statements for \TDKR and \TDK.

\begin{theorem}\label{theorem:knapsack:with:rotations:hardness:appendix}
\TDKR is \W{1}-hard.
\end{theorem}

\begin{proof}
We give a polynomial-time parameterized reduction from \MSS to \TDKR with output parameter $k'=O(k^2)$. By Lemma~\ref{lemma:mssk:hardness}, this establishes \W{1}-hardness of \TDKR.

\emph{Construction.}
Let $(x_1,\ldots,x_{m},t,k)$ be an instance of \MSS. W.l.o.g. we may assume that $4\leq k \leq m$ and that $x_i< t$ for all $i\in[m]$. Furthermore, as solutions may select the same integer multiple times, we may assume that all the $x_i$ are pairwise different. 

Throughout, we take a knapsack to be an $N$ by $N$ square with coordinate $(0,0)$ in the bottom left corner and $(N,N)$ at top right. The first coordinate of any point in the knapsack measures the horizontal (left-right) distance from the point to $(0,0)$; the second coordinate measure the vertical (up-down) distance from $(0,0)$. All items in the following construction are given such that their sizes reflect their \emph{intended rotation} in a solution, i.e., heights refers to vertical dimensions and widths to horizontal dimensions. 

We begin by constructing an instance of \TDK. Throughout, for an item $R$, we will use $\HEIGHT(R)$ and $\WIDTH(R)$ denote its height and width. The instance of \TDK is defined as follows:
\begin{itemize}\itemsep0pt
 \item We define constants 
 \begin{align*}
 S := k^2\cdot t \hspace{3cm}
 L := k^2\cdot S=k^4\cdot t. 
 \end{align*}
 (The specific values will not be important so long as $k^2\cdot t\leq S$ and $k^2 \cdot S \leq L$. Intuitively, the identifiers are chosen to mean \emph{small} and \emph{large}.)
 \item The knapsack has height and width both equal to 
 \begin{align}
 N:=k\cdot L+ (2k-1)\cdot S + (2k-1)\cdot t. \label{math:kr:knapsacksize}
 \end{align}
 \item For each $i\in [m]$ we construct $k^2$ items $R(i,1),\ldots,R(i,{k^2})$ with 
 \begin{align}
 \HEIGHT(R(i,j)) &= L+S+ x_i \label{math:kr:tileheight}\\
 \WIDTH(R(i,j)) &= L+S+ 2t- x_i. \label{math:kr:tilewidth}
 \end{align}
 We call these items \emph{tiles}. We say that each tile $R(i,\cdot)$ corresponds to the number $x_i$ from the input that it was constructed for. Since the $x_i$ are pairwise different, the $x_i$ corresponding to any tile can be easily read off from both height and width. We point out that all tiles have height strictly between $L+S$ and $L+S+t$ and width strictly between $L+S+t$ and $L+S+2t$.

 \item We add $p:= k\cdot (k-1)$ items $T(1),\ldots,T(p)$ with height $L$ and width $S$. We call these the \emph{thin items}.
 \item We add $p$ items $F(1),\ldots,F(p)$ with height $S$ and width $L$. We call these the \emph{flat items}.
 \item We add a single (very flat and very wide) item of height $(2k-2)\cdot t$ and width $N$, which we call the \emph{bar}. 
 \item The created instance has a target value of $k'=k^2+2p+1$. (The intention is to pack all thin and all flat items, the bar, and exactly $k^2$ tiles.)
\end{itemize}
This completes the construction. Clearly, all necessary computations can performed in polynomial time. The parameter value $k'=k^2+2p+1$ is upper bounded by $O(k^2)$. It remains to prove correctness.

\emph{Correctness.}
We need to prove that the instance $(x_1,\ldots,x_{m},t,k)$ is yes for \MSS if and only if the constructed instance is yes for \TDKR.

$\Longleftarrow:$ Assume that the created instance is yes for \TDKR, i.e., that it has a packing with $k'=k^2+2p+1$ items and fix any such packing. Observe that the packing must contain at least $k^2$ tiles as there are only $2p+1$ items that are not tiles. We will show that the packing uses exactly $k^2$ tiles, the $2p$ thin/flat items, and the bar. It is useful to recall that tiles have height and width both greater than $L+S$ no matter whether they are rotated.

Consider the effect of placing $k$ vertical lines in the knapsack at horizontal coordinates $L-1,2\cdot(L-1),\ldots,k\cdot (L-1)$. We first observe that these lines must necessarily intersect
all tiles of the packing because each of them has width at least $L$: The distance between any two consecutive lines is $L-1$, same as the distance from the left border of the knapsack to the first line. The distance from the $k$th vertical line to the right border is also strictly less than $L$:
\[
N-k\cdot (L-1)= k + (2k-1)\cdot S + (2k-1)\cdot t < S + (2k-1)\cdot S + S = (2k+1)\cdot S < L
\]

Observe that no line can intersect more than $k$ tiles: Any two tiles of the packing may not overlap and may in particular not share their intersection with any line. Since each line has length $N$ and each intersection with a tile has length greater than $L$, there can be at most $k$ tiles intersected by any line as $N<(k+1)\cdot L$: 
\[
N=k\cdot L + (2k-1)\cdot S + (2k-1)\cdot t < k\cdot L + 4k\cdot S \leq (k+1) \cdot L 
\]

Overall, this means that the packing contains at most $k^2$ tiles: There are $k$ lines that intersect all tiles of the packing, each of them intersecting at most $k$. By our earlier observation, this implies that the packing contains exactly $k^2$ tiles in addition to all $2p$ flat/thin items. Moreover, each line intersects exactly $k$ tiles and no two lines intersect the same tile.

Let us now check how the vertical lines and the flat and thin items interact. Clearly, each both flat as well as rotated thin items have width $L$ and height $S$. Accordingly, each flat and each rotated thin item must be intersected by at least one of the $k$ vertical lines. We already know that a total length of at least $k\cdot (L+S)$ of each line is occupied by the $k$ tiles that the line intersects. This leaves at most a length of
\[
N-k\cdot (L+S)= (k-1)\cdot S + (2k-1)\cdot t < k\cdot S
\]
for intersecting flat and rotated thin items, and allows for intersecting at most $k-1$ of them. (Again, no two items can share their intersection with the line.) Thus, there are at most $p=k\cdot (k-1)$ of the flat and rotated thin items in the packing.

Before analyzing the vertical lines further, let us perform an analogous argument for $k$ horizontal lines with vertical coordinates $L-1, 2\cdot(L-1),\ldots,k\cdot (L-1)$ and their intersection with tiles and flat/thin items. It can be verified that each of them similarly intersects exactly $k$ tiles and that no tile is intersected twice. The argument for flat and thin items is analogous as well, except that we now reason about rotated flat and (non-rotated) thin items, which have height $L$ and width $S$; we find that there are at most $p$ such items and that each horizontal line intersects at most $k-1$ of them. Since in total there must be $2p$ flat and thin items, this implies that both sets of lines (horizontal and vertical) intersect $p$ of these items each. Since flat and thin items can be swapped freely, we may assume that none of these items are rotated, and that the vertical lines intersect the $p$ flat items and the horizontal lines intersect the $p$ thin items.

We know now that the packing contains exactly $k^2$ tiles as well as the $p$ flat and the $p$ thin items. Thus, to get a total of $k'=k^2+2p+1$ items, it must also contain the bar, which has height $N$ and width $(2k-2)\cdot t$. W.l.o.g., we may assume that the bar is not rotated, or else we could rotate the entire packing.\footnote{Note that this does not change the fact that no rotation of flat and thin items is required since there are still $p$ each such items in either rotation. In other words, if there is a feasible packing then there is one where neither the bar nor the thin/flat items are rotated. We will show that in such a packing also the tiles are not rotated.} It follows that all vertical lines intersect the bar due to its width of $N$, which matches the width of the knapsack.

Let us now analyze both vertical and horizontal lines further. The goal is to obtain inequalities on the values $x_i$ that go into the construction of the tiles; up to now we have only used that they are fairly large.
We know that each vertical line intersects $k$ tiles, $k-1$ flat items, and the bar. Let $h_1,\ldots,h_k$ denote the heights of the tiles (ordered arbitrarily) and recall that each flat item has height $S$ while the bar has height $(2k-2)\cdot t$. Since all intersections with the line are disjoint and the line has length $N$ (equaling the height of the knapsack), we get that
\begin{align}
N \geq h_1 + \ldots + h_k + (k-1)\cdot S + (2k-2)\cdot t. \label{math:kr:verticallinebound}
\end{align}
At this point, in order to plug in values for the $h_i$, it is important whether any of the tiles are rotated; we will show that having at least one rotated tile causes a violation of (\ref{math:kr:verticallinebound}). To this end, recall that (non-rotated) tiles have heights strictly between $L+S$ and $L+S+t$ and widths strictly between $L+S+t$ and $L+S+2t$. Thus, if at least one tile is rotated then it has height greater than $L+S+t$, rather than the weaker bound of greater than $L+S$. Using this, the right-hand side of (\ref{math:kr:verticallinebound}) can be lower bounded by
\begin{align*}
RHS &> (k-1)\cdot (L+S) + (L+S+t) + (k-1)\cdot S + (2k-2)\cdot t\\
&= k\cdot L + (2k-1)\cdot S + (2k-1)\cdot t\\
&= N,
\end{align*}
contradicting (\ref{math:kr:verticallinebound}).
Thus, none of the tiles intersected by the vertical line can be rotated. Since each tile is intersected by a vertical line, it follows that no tiles can be rotated and we can analyze the lines using the sizes as given in (\ref{math:kr:tileheight}) and (\ref{math:kr:tilewidth}). Let us return to replacing the values $h_i$ in (\ref{math:kr:verticallinebound}).
Recall that the height of a tile is equal to $L+S+ x_i$ where $x_i$ is the corresponding integer from the input to the initial \MSS instance. Thus, if the $i$th intersected tile corresponds to input integer $y_i\in\{x_1,\ldots,x_{m}\}$ then by (\ref{math:kr:tileheight}) we have
\[
h_i = L + S + y_i.
\]
Plugging this into (\ref{math:kr:verticallinebound}) yields
\begin{align*}
N  &\geq \sum_{i=1}^k (L + S + y_i) + (k-1)\cdot S + (2k-2)\cdot t\\
 &= k\cdot L + (2k-1)\cdot S + (2k-2)\cdot t + \sum_{i=1}^k y_i.
\end{align*}
Using $N=k\cdot L+(2k-1)\cdot S + (2k-1)\cdot t$ we immediately get
\begin{align}
t\geq \sum_{i=1}^k y_i. \label{math:kr:verticalbound}
\end{align}

Let us apply the same argument to the horizontal lines: Each such line intersects $k$ tiles and $k-1$ thin items. Let $w_1,\ldots,w_k$ denote the widths of the tiles (ordered arbitrarily) and recall that each thin item has width $S$. As intersections with the line are disjoint and its length is $N$, we get that
\begin{align}
N \geq w_1 + \ldots + w_k + (k-1)\cdot S. \label{math:kr:horizontallinebound}
\end{align}
We already know that none of the tiles are rotated. We recall that the width of a tile corresponding to input integer $x_i$ is equal to $L+S+2t-x_i$ (\ref{math:kr:tilewidth}).
Thus, if the $i$th intersected tile corresponds to input integer $z_i\in\{x_1,\ldots,x_{m}\}$ then we have
\[
w_i = L + S + 2t - z_i.
\]
Plugging this into (\ref{math:kr:horizontallinebound}) yields
\begin{align*}
N & \geq \sum_{i=1}^k (L + S + 2t - z_i) + (k-1)\cdot S = k\cdot L + (2k-1) \cdot S + 2k\cdot t - \sum_{i=1}^k z_i.
\end{align*}
Using $N= k\cdot L + (2k-1)\cdot S + (2k-1)\cdot t$ this simplifies to
\begin{align}
-t  \geq - \sum_{i=1}^k z_i 
\quad \iff \quad t \leq \sum_{i=1}^k z_i. \label{math:kr:horizontalbound}
\end{align}

Recall that there are exactly $s=k^2$ tiles in the packing, and let $x_{i_1},\ldots,x_{i_s}$ be the corresponding values from the input. They are partitioned into $k$ groups of size $k$ each by the $k$ vertical lines, and again by the $k$ horizontal lines. (I.e., the group corresponding to a line is the set of those $k$ tiles that are intersected by the line.) The grouping by vertical lines yields $k$ inequalities of form (\ref{math:kr:verticalbound}), with each $x_{i_j}$ appearing in exactly one of them. The grouping by horizontal lines yields $k$ inequalities of form (\ref{math:kr:horizontalbound}), and again each $x_{i_j}$ appears in exactly one of them. (That is, values may be repeated but the formal variable $x_{i_j}$ appears exactly once.) It follows that all the inequalities must be fulfilled with equality, so that $\sum_{j=1}^s x_{i_j}= k\cdot t$ holds.

Picking any single inequality of form (\ref{math:kr:verticalbound}) for any vertical line that intersects tiles corresponding to input integers $y_1,\ldots,y_k\in\{x_1,\ldots,x_{m}\}$ we get
$
t = \sum_{i=1}^k y_i.
$
In other words, there is a selection of $k$ input values $y_1,\ldots,y_k\in\{x_1,\ldots,x_{m}\}$, possibly with repetition, that sums to exactly $t$. Thus, the initial instance for \MSS is a yes-instance, as required. This completes the first part of the correctness proof.

$\Longrightarrow:$
For the converse, assume that the input \MSS instance has a solution, i.e., that we can select $k$ numbers $y_1,\ldots,y_k\in\{x_1,\ldots,x_{m}\}$, allowing repetition, such that
$
t = \sum_{i=1}^k y_i.
$
We will show how to get a packing of $k'=k^2 + 2p+1$ items for the created \TDK instance without using any rotations. Concretely, we will be using only flat/thin items, the bar, and the tiles that correspond to the numbers $y_1,\ldots,y_k$. (Recall that numbers may be repeated, which is why we created $k^2$ items per input number during the construction.)

We will construct a packing that arranges $k^2$ tiles in roughly grid form, i.e., we use $k$ by $k$ tiles. Between the tiles we will insert $2p$ flat/thin items and the bar is added at the top of the knapsack. Let us denote the $k^2$ tiles in the packing by $R_{a,b}$ with $a,b\in [k]$. Concretely, we use the following tiles from the construction:
\begin{itemize}
 \item $R_{1,1},R_{2,2},\ldots,R_{k-1,k-1},R_{k,k}$ are tiles corresponding to $y_1$, i.e., they have height $L+S +y_1$ and width $L+S+2t-y_1$.
 \item $R_{2,1},R_{3,2},\ldots,R_{k,k-1},R_{1,k}$ are tiles corresponding to $y_2$, and so on.
 \item $R_{k,1},R_{1,2},\ldots,R_{k-2,k-1},R_{k-1,k}$ are tiles corresponding to $y_k$.
\end{itemize}
More formally, item $R_{a,b}$ is a tile corresponding to $y_i$, where $i=1+((a-b) \mod k)$, and accordingly has $\HEIGHT(R_{a,b})=L+S+y_i$ and $\WIDTH(R_{a,b})= L + S + 2t-y_i$. This yields the required property that for each $a\in[k]$ the items $R_{a,1},\ldots,R_{a,k}$ contain tiles corresponding to all numbers $y_1,\ldots,y_k$ (and correctly contain multiple copies for numbers that appear more than once). The same holds for items $R_{1,b},\ldots,R_{k,b}$ for all $b\in[k]$.

We use $\HEIGHT(R_{i,j})$ and $\WIDTH(R_{i,j})$ to refer to height and width of tile $R_{i,j}$.
We use $\LEFT(R)$, $\RIGHT(R)$, $\TOP(R)$, and $\BOTTOM(R)$ to specify the coordinates of any item in our packing, i.e., for the $k^2$ tiles, the $2p$ flat/thin items, and the bar. The coordinates for tiles are chosen as
\begin{align*}
\LEFT(R_{a,b}) &= (a-1) \cdot S + \sum_{i=1}^{a-1} \WIDTH(R_{i,b}),\\ 
\RIGHT(R_{a,b}) &= (a-1) \cdot S + \sum_{i=1}^a \WIDTH(R_{i,b}),\\
\BOTTOM(R_{a,b}) &= (b-1) \cdot S + \sum_{i=1}^{b-1} \HEIGHT(R_{a,i}),\\ 
\TOP(R_{a,b}) &= (b-1) \cdot S + \sum_{i=1}^b \HEIGHT(R_{a,i}).
\end{align*}

Let us first check some basic properties of these coordinates:
\begin{itemize}
 \item We observe that each tile is assigned coordinates that match its size, i.e., $\WIDTH(R_{a,b})=\RIGHT(R_{a,b})-\LEFT(R_{a,b})$ and $\HEIGHT(R_{a,b})=\TOP(R_{a,b})-\BOTTOM(R_{a,b})$.
 \item All coordinates lie inside the knapsack. Clearly, all coordinates are non-negative and it suffices to give upper bounds for $\TOP(R_{a,k})$ and $\RIGHT(R_{k,b})$. Recall that by construction each set of tiles $R_{a,1},\ldots,R_{a,k}$ contains tiles corresponding to all numbers $y_1,\ldots,y_k$, and same for $R_{1,b},\ldots,R_{k,b}$. Thus we get
 \begin{align*}
  \RIGHT(R_{k,b}) &= (k-1)\cdot S + \sum_{i=1}^k \WIDTH(R_{i,b})\\
  &= (k-1)\cdot S + \sum_{i=1}^k (L+S+2t-y_i)\\
  &= k\cdot L + (2k-1)\cdot S + 2k\cdot t - \sum_{i=1}^k y_i\\
  &= k\cdot L + (2k-1)\cdot S + (2k-1)\cdot t\\
  &= N.
 \end{align*}
 Similarly, we get
 \begin{align*}
  \TOP(R_{a,k}) &= (k-1)\cdot S + \sum_{i=1}^k \HEIGHT(R_{a,i})\\ 
  &= (k-1)\cdot S + \sum_{i=1}^k (L+S+y_i)\\
  &= k\cdot L + (2k-1)\cdot S + \sum_{i=1}^k y_i \\
  &= k\cdot L + (2k-1)\cdot S + t \\
  &= N - (2k-2)\cdot t.
 \end{align*}
 We will later use the gap of $(2k-2)\cdot t$ between $N$ and $N-(2k-2)\cdot t$ to place the bar item, as its height exactly matches the gap.
 \item For any tile $R_{a,b}$ the possible coordinates fall into very small intervals, using that all heights and widths of tiles lie strictly between $L+S$ and $L+S+2t$. We show this explicitly for $\LEFT(R_{a,b})$:
 \begin{align*}
  \LEFT(R_{a,b})&=(a-1)\cdot S + \sum_{i=1}^{a-1}\WIDTH(R_{i,b})\\
  \LEFT(R_{a,b})&>(a-1)\cdot S + \sum_{i=1}^{a-1} (L+S) =(a-1)\cdot L+ (2a-2)\cdot S\\
  \LEFT(R_{a,b})&<(a-1)\cdot S + \sum_{i=1}^{a-1}(L+S+2t) =(a-1)\cdot L + (2a-2)\cdot S + (2a-2)\cdot t\\
  &<(a-1)\cdot L + (2a-1)\cdot S
 \end{align*}
 In this way, we get the following intervals for $\LEFT(R_{a,b})$, $\RIGHT(R_{a,b})$, $\BOTTOM(R_{a,b})$, and $\TOP(R_{a,b})$. (Note that we sacrifice the possibility of tighter bounds in order to get the same simple form of bound for $\TOP$ and $\RIGHT$ and for $\BOTTOM$ and $\LEFT$.)
 \begin{align}
  (a-1)\cdot L+ (2a-2)\cdot S &&&<&& \LEFT(R_{a,b}) &&<&& (a-1)\cdot L + (2a-1)\cdot S \label{math:kr:leftinterval}\\
  a\cdot L + (2a-1)\cdot S &&&<&& \RIGHT(R_{a,b}) &&<&& a\cdot L + 2a\cdot S \label{math:kr:rightinterval}\\
  (b-1)\cdot L+ (2b-2)\cdot S &&&<&& \BOTTOM(R_{a,b}) &&<&& (b-1)\cdot L + (2b-1)\cdot S \label{math:kr:bottominterval}\\
  b\cdot L + (2b-1)\cdot S &&&<&& \TOP(R_{a,b}) &&<&& b\cdot L + 2b\cdot S \label{math:kr:topinterval}
 \end{align}
\end{itemize}

We can now easily verify that no two tiles $R_{a,b}$ and $R_{c,d}$ overlap if $(a,b)\neq (c,d)$. If $a\neq c$ then we may assume w.l.o.g. that $a<c$ (and hence $a\leq c-1$). Using (\ref{math:kr:topinterval}) and (\ref{math:kr:bottominterval}) we get
\begin{align*}
\RIGHT(R_{a,b}) &< a\cdot L + 2a\cdot S \leq (c-1)\cdot L + (2c-2)\cdot S < \LEFT(R_{c,d}).
\end{align*}
Thus, $R_{a,b}$ and $R_{c,d}$ do not overlap if $a\neq c$. If instead $a=c$ then we must have $b\neq d$ and, w.l.o.g., $b<d$ (and hence $b\leq d-1$). Thus we have
\begin{align*}
\TOP(R_{a,b}) &< b\cdot L + 2b\cdot S \leq (d-1)\cdot L + (2d-2)\cdot S < \BOTTOM(R_{c,d}).
\end{align*}
Thus, no two tiles $R_{a,b}$ and $R_{c,d}$ with $(a,b)\neq (c,d)$ overlap.

We will now specify coordinates for the $p$ flat and the $p$ thin items. For this purpose the intervals for coordinates of the tiles (\ref{math:kr:leftinterval})--(\ref{math:kr:topinterval}) are highly useful. For thin items, there will always be two adjacent tiles, to the left and to the right, and we use the intervals to get top and bottom coordinates. For flat items the situation is the opposite; there are adjacent tiles on the top and bottom sides and we use the intervals to get left and right coordinates. Recall that thin items have height $L$ and width $S$, whereas flat items have height $S$ and width $L$.

We denote the $p$ thin items by $T_{a,b}$ with $a\in[k-1]$ and $b\in[k]$; we choose coordinates as follows:
\begin{align}
\LEFT(T_{a,b}) &= \RIGHT(R_{a,b}) = (a-1)\cdot S + \sum_{i=1}^a \WIDTH(R_{i,b}) \label{math:kr:thinleft}\\
\RIGHT(T_{a,b}) &= \LEFT(R_{a+1,b}) = a\cdot S + \sum_{i=1}^a \WIDTH(R_{i,b}) \label{math:kr:thinright}\\
\BOTTOM(T_{a,b}) &= (b-1)\cdot L+ (2b-1)\cdot S \label{math:kr:thinbottom}\\
\TOP(T_{a,b}) &= b\cdot L + (2b-1) \cdot S \label{math:kr:thintop}
\end{align}
Clearly, the coordinates match the dimension of $T_{a,b}$.

We denote the $p$ flat items by $F_{a,b}$ with $a\in[k]$ and $b\in[k-1]$, and we use the following coordinates:
\begin{align}
\LEFT(F_{a,b}) &= (a-1) \cdot L + (2a-1)\cdot S \label{math:kr:flatleft}\\
\RIGHT(F_{a,b}) &= a\cdot L + (2a-1)\cdot S \label{math:kr:flatright}\\
\BOTTOM(F_{a,b}) &= \TOP(R_{a,b}) = (b-1)\cdot S + \sum_{i=1}^b \HEIGHT(R_{a,i}) \label{math:kr:flatbottom}\\
\TOP(F_{a,b}) &= \BOTTOM(R_{a,b+1}) = b\cdot S + \sum_{i=1}^b \HEIGHT(R_{a,i}) \label{math:kr:flattop}
\end{align}
Clearly, the coordinates match the dimension of $F_{a,b}$. It remains to show that there is no overlap between any of the items placed so far (all except the bar), recalling that intersections between tiles are already ruled out: It remains to consider (1) tile-flat, (2) tile-thin, (3) flat-flat, (4) flat-thin, and (5) thin-thin overlaps.

(1) There are no overlaps between any tile $R_{a,b}$ and any flat item $F_{c,d}$:
\begin{itemize}
 \item If $a<c$ then $a\leq c-1$ and using (\ref{math:kr:rightinterval}) and (\ref{math:kr:flatleft}) we get
 \begin{align*}
  \RIGHT(R_{a,b}) &< a\cdot L + 2a\cdot S \leq (c-1)\cdot L + (2c-2)\cdot S\\
  &< (c-1)\cdot L + (2c-1)\cdot S = \LEFT(F_{c,d}).
 \end{align*}
 \item If $a>c$ then $c\leq a-1$ and using (\ref{math:kr:flatright}) and (\ref{math:kr:leftinterval}) we get
 \begin{align*}
  \RIGHT(F_{c,d}) &= c\cdot L + (2c-1)\cdot S \leq (a-1)\cdot L + (2a-3) \cdot S\\
  &<(a-1)\cdot L + (2a-2)\cdot S < \LEFT(R_{a,b}).
 \end{align*}
 \item If $a=c$ and $b\leq d$ then using (\ref{math:kr:flatbottom}) we get
 \begin{align*}
  \TOP(R_{a,b}) &\leq \TOP(R_{a,d})= \TOP(R_{c,d}) = \BOTTOM(F_{c,d}).
 \end{align*}
 \item If $a=c$ and $b>d$ then $d+1\leq b$ and using (\ref{math:kr:flattop}) we get
 \begin{align*}
  \TOP(F_{c,d}) &= \BOTTOM(R_{c,d+1}) = \BOTTOM(R_{a,d+1}) \leq \BOTTOM(R_{a,b}).
 \end{align*}
\end{itemize}
Thus, in all four cases there is no overlap, as claimed.

(2) There are no overlaps between any tile $R_{a,b}$ and any thin item $T_{c,d}$:
\begin{itemize}
 \item If $b<d$ then $b\leq d-1$ and using (\ref{math:kr:topinterval}) and (\ref{math:kr:thinbottom}) we get
 \begin{align*}
  \TOP(R_{a,b}) &< b\cdot L + 2b\cdot S \leq (d-1)\cdot L + (2d-2)\cdot S\\
  &< (d-1)\cdot L + (2d-1)\cdot S =\BOTTOM(T_{c,d}).
 \end{align*}
 \item If $b>d$ then $d\leq b-1$ and using (\ref{math:kr:thintop}) and (\ref{math:kr:bottominterval}) we get
 \begin{align*}
  \TOP(T_{c,d}) &= d\cdot L + (2d-1)\cdot S \leq (b-1)\cdot L + (2b-3)\cdot S\\
  &< (b-1)\cdot L + (2b-2)\cdot S < \BOTTOM(R_{a,b}).
 \end{align*}
 \item If $b=d$ and $a\leq c$ then using (\ref{math:kr:thinleft}) we get
 \begin{align*}
  \RIGHT(R_{a,b}) &\leq \RIGHT(R_{c,b}) =\RIGHT(R_{c,d}) = \LEFT(T_{c,d}).
 \end{align*}
 \item If $b=d$ and $a>c$ then $c+1\leq a$ and using (\ref{math:kr:thinright}) we get
 \begin{align*}
  \RIGHT(T_{c,d}) = \LEFT(R_{c+1,d}) = \LEFT(R_{c+1,b}) \leq \LEFT(R_{a,b}).
 \end{align*}
\end{itemize}
Thus, in all four cases there is no overlap, as claimed.

(3) There are no overlaps between any two flat items $F_{a,b}$ and $F_{c,d}$ when $(a,b)\neq (c,d)$:
\begin{itemize}
 \item If $a\neq c$ then, w.l.o.g., $a < c$ (and hence $a\leq c-1$) and using (\ref{math:kr:flatright}) and (\ref{math:kr:flatleft}) we get
 \begin{align*}
  \RIGHT(F_{a,b})&= a\cdot L + (2a-1)\cdot S \leq (c-1)\cdot L + (2c-3)\cdot S\\
  &< (c-1)\cdot L + (2c-1)\cdot S = \LEFT(F_{c,d}).
 \end{align*}
 \item If $a=c$ then, due to $(a,b)\neq (c,d)$, we have $b\neq d$ and, w.l.o.g., $b < d$. Thus, $b+1\leq d$ and using (\ref{math:kr:flattop}) and (\ref{math:kr:flatbottom}) we get
 \begin{align*}
  \TOP(F_{a,b}) &= \BOTTOM(R_{a,b+1}) \leq \BOTTOM(R_{a,d}) \leq \TOP(R_{a,d}) = \TOP(R_{c,d})\\
  &= \BOTTOM(F_{c,d}).
 \end{align*}
\end{itemize}
Thus, in both cases there is no overlap, as claimed.

(4) There are no overlaps between any flat item $F_{a,b}$ and any thin item $T_{c,d}$:
\begin{itemize}
 \item If $a\leq c$ then using (\ref{math:kr:flatright}), (\ref{math:kr:rightinterval}), and (\ref{math:kr:thinleft}) we get
 \begin{align*}
  \RIGHT(F_{a,b}) &= a\cdot L + (2a-1)\cdot S \leq c\cdot L + (2c-1) \cdot S < \RIGHT(R_{c,d}) = \LEFT(T_{c,d}).
 \end{align*}
 \item If $a>c$ then $c\leq a-1$ and using (\ref{math:kr:thinright}), (\ref{math:kr:leftinterval}), and (\ref{math:kr:flatleft}) we get
 \begin{align*}
  \RIGHT(T_{c,d}) &= \LEFT(R_{c+1,d}) < c\cdot L + (2c+1)\cdot S \leq (a-1)\cdot L + (2a-1)\cdot S \\
  &= \LEFT(F_{a,b}).
 \end{align*}
\end{itemize}
Thus, in both cases there is no overlap, as claimed.

(5) There are no overlaps between any two thin items $T_{a,b}$ and $T_{c,d}$ when $(a,b)\neq (c,d)$:
\begin{itemize}
 \item If $b\neq d$ then, w.l.o.g., $b < d$ (and hence $b\leq d-1$) and using (\ref{math:kr:thintop}) and (\ref{math:kr:thinbottom}) we get
 \begin{align*}
  \TOP(T_{a,b}) &= b\cdot L + (2b-1)\cdot S \leq (d-1)\cdot L + (2d-3)\cdot S\\
  &< (d-1)\cdot L + (2d-1)\cdot S = \BOTTOM(T_{c,d}).
 \end{align*}
 \item If $b=d$ then, due to $(a,b)\neq (c,d)$, we have $a\neq c$ and, w.l.o.g., $a < c$. Thus, $a+1\leq c$ and using (\ref{math:kr:flatright}) and (\ref{math:kr:flatleft}) we get
 \begin{align*}
  \RIGHT(T_{a,b}) &= \LEFT(R_{a+1,b}) \leq \LEFT(R_{c,b}) \leq \RIGHT(R_{c,b}) =\RIGHT(R_{c,d}) \\
  &=\LEFT(T_{c,d}).
 \end{align*}
\end{itemize}
Thus, in both cases there is no overlap, as claimed. Overall, we find that there are no overlap between any pair of items placed so far. It remains to add the bar to complete our packing.

We already observed earlier that $\TOP(R_{a,k}) = N-(2k-2)\cdot t$. Similarly, using (\ref{math:kr:flattop}) we get 
\begin{align*}
\TOP(F_{a,b}) &= \BOTTOM(R_{a,b+1}) \leq \BOTTOM(R_{a,k}) \leq \TOP(R_{a,k}) \leq N-(2k-2)\cdot t
\end{align*}
for all $a\in[k]$ and $b\in[k-1]$. In the same way, using (\ref{math:kr:thintop}) we get
\begin{align*}
\TOP(T_{a,b}) &= b\cdot L + (2b-1)\cdot S \leq k\cdot L + (2k-1)\cdot S < N-(2k-2)\cdot t
\end{align*}
for all $a\in[k-1]$ and $b\in [k]$, recalling that $N=k\cdot L+(2k-1)\cdot S+(2k-1)\cdot t$. Thus, we can place the bar $B$ of height $(2k-2)\cdot t$ and width $N$ at the top of the knapsack without causing overlaps; formally, its coordinates are as follows.
\begin{align*}
\LEFT(B) = 0 \hspace{0.9cm}
\RIGHT(B) = N \hspace{0.9cm}
\BOTTOM(B) = N-(2k-2)\cdot t \hspace{0.9cm}
\TOP(B) = N
\end{align*}

Overall, we have placed $k^2+2p+1$ items without overlap. Thus, the constructed instance of \TDK is a yes-instance, as required. This completes the proof.
\end{proof}

\begin{corollary}\label{corollary:knapsack:hardness:appendix}
The \TDK problem is \W{1}-hard.
\end{corollary}

\begin{proof}
We can use the same construction as in the proof of Theorem~\ref{theorem:knapsack:with:rotations:hardness:appendix} to get a parameterized reduction from \MSS to \TDK.

If the constructed instance is yes for \TDK then it is also yes for \TDKR, as the same packing of $k'=k^2+2p+1$ items can be used. As showed earlier, the latter implies that the input instance is yes for \MSS. Conversely, if the input instance is yes for \MSS then we already showed that there is a feasible packing to show that the constructed instance is yes for \TDKR. Since the packing did not require rotation of any items, it is also a feasible solution showing that the instance is yes for \TDK.
\end{proof}

\end{document}